\documentclass[sigplan,10pt,nonacm]{acmart}
\setcopyright{none}             
\usepackage{xcolor}
\usepackage{relsize}

\usepackage{anyfontsize}

\usepackage{oz}
\usepackage{nameref}

\usepackage{subfigure}
\usepackage{refcal}
\usepackage{url}
\usepackage[strings]{underscore}

\usepackage{rcp}
\usepackage{xspace}

\usepackage{datetime}
\usepackage{rcoms}
\usepackage{opsem}
\usepackage{multicol}

\usepackage{stmaryrd}

\usepackage{ctable}

\usepackage{wmm}
\usepackage{soundness}
\usepackage{counters}

\usepackage{enumitem}

\renewcommand{\ltiff}[1]{~\mbox{iff~#1}}

\renewcommand{\ruledefNamed}[4]{
	\begin{minipage}{#1}
	\begin{equation}
		\label{rule:#3}
			#4
	\end{equation}
	\end{minipage}
}

\renewcommand{\mmdefTitle}{Model}

\newcommand{\shorteqn}[2]{
	\begin{minipage}[b]{#1}
	\begin{equation}
	#2
	\end{equation}
	\end{minipage}
}

\newcommand{\htripp}[2]{\htrip{p}{#1}{#2}}
\newcommand{\htrippq}[1]{\htripp{#1}{q}}

\renewcommand{\step}[1]{$#1$}
\renewcommand{\trans}[1]{$#1$&}
\newcommand{\explanation}[1]{& #1 \\}

\definecolor{dkgreen}{rgb}{0,0.6,0}
\definecolor{gray}{rgb}{0.5,0.5,0.5}
\definecolor{mauve}{rgb}{0.58,0,0.82}

\renewcommand{\csc}{\cunder{\SCmm}}
\newcommand{\ctso}{\cunder{\TSOmm}}

\renewcommand{\Meaning}[1]{[\!\![ #1 ]\!]}
\renewcommand{\Meaning}[1]{\llbracket #1 \rrbracket}

\newcommand{\effword}{{\sf eff}}
\newcommand{\eff}[1]{\effword(#1)}
\newcommand{\effa}{\eff{\aca}}
\newcommand{\effc}{\eff{\cmdc}}
\newcommand{\State}{\Sigma}

\renewcommand {\refth}[1] {Thm.~\ref{theorem:#1}}
\renewcommand {\refth}[1] {Theorem~\ref{theorem:#1}}

\renewcommand{\refeq}{=}

\newcommand{\bigunion}{{\textstyle\bigcup}}
\newcommand{\bigintersect}{{\textstyle\bigcap}}
\newcommand{\bigintU}[1]{\underset{#1}{\bigintersect}}
\newcommand{\bigintUn}{\bigintU{n}}

\let\Oldcat\cat
\renewcommand{\cat}{\mathsmaller{\mathsmaller{\Oldcat}}}

\renewcommand{\True}{True}
\renewcommand{\False}{False}

\newcommand{\isLoad}[1]{\mathsf{isLoad}(#1)}
\newcommand{\isStore}[1]{\mathsf{isStore}(#1)}
\newcommand{\isLoada}{\isLoad{\aca}}
\newcommand{\isLoadb}{\isLoad{\acb}}
\newcommand{\isStorea}{\isStore{\aca}}

\newcommand{\isReg}[1]{\mathsf{isReg}(#1)}
\newcommand{\isRega}{\isReg{\aca}}

\newcommand{\ppseq}[1]{\mathbin{\raisebox{0.0pt}{$\overset{\raisebox{0.0pt}{$\mathsmaller{\mathsmaller{#1}}$}}{\semicolon}$}}}
\newcommand{\ppseqm}{\ppseq{\mm}}
\newcommand{\ppseqc}{\mathbin{\ppseq{}}}
\newcommand{\ppseqs}{\ppseq{\SCmm}}
\newcommand{\ppseqt}{\ppseq{\TSOmm}}
\newcommand{\ppseqA}{\ppseq{\ARMmm}}
\newcommand{\ppseqG}{\ppseq{\Gmm}}
\newcommand{\ppseqGz}{\ppseq{\Gzmm}}
\newcommand{\ppseqRV}{\ppseq{\RVmm}}

\renewcommand{\scomp}{\mathbin{.{}\semicolon}}

\newcommand{\guardf}{\guard{f}}

\renewcommand{\reg}[2]{\mathtt{REG~#1, #2}}

\renewcommand{\refeqn}[1]{(\ref{eqn:#1})}
\renewcommand{\refeqns}[2]{\refeqn{#1} and~\refeqn{#2}}
\renewcommand{\WHbc}{\While b~\Do~c}

\renewcommand{\reftable}[1]{Table~\ref{table:#1}}
\renewcommand{\refsect}[1]{Sect.~\ref{sect:#1}}

\newcommand{\iterate}[2]{#1^\omega_#2}
\newcommand{\iteratec}[1]{\iterate{\cmdc}{#1}}
\newcommand{\iteratecm}{\iteratec{\mm}}

\newcommand{\finiterate}[3]{#1^{#3}_{#2}}
\newcommand{\finiteratec}[2]{\finiterate{\cmdc}{#1}{#2}}
\newcommand{\finiteratecm}[1]{\finiteratec{\mm}{#1}}
\newcommand{\finiteratecmn}{\finiteratecm{n}}

\newcommand{\labellaw}[1]{\label{law:#1}}
\newcommand{\labeltable}[1]{\label{table:#1}}
\newcommand{\labeldefn}[1]{\label{defn:#1}}

\renewcommand{\refmm}[1]{{\bf M\ref{mm:#1}}}
\newcommand{\refMM}[1]{{\bf Model~\ref{mm:#1}}}

\newcommand{\scat}{\cat}

\renewcommand{\refrule}[1]{(\ref{rule:#1})}
\renewcommand{\refrules}[2]{(\ref{rule:#1}, \ref{rule:#2})}

\renewcommand{\store}[2]{(#1 , #2)}
\renewcommand{\storex}[1]{\store{x}{#1}}
\renewcommand{\storexv}{\storex{v}}

\newcommand{\tsofence}{\mathtt{mfence}}
\newcommand{\tsolfence}{\mathtt{lfence}}

\def\atomicsep{~,~}

\renewcommand{\LOCALSWord}{\mathbf{local}\xspace}
\renewcommand{\GLOBALSWord}{\mathbf{shared}\xspace}

\renewcommand{\prefix}[2]{#1 \cbef #2}
\renewcommand{\prefixa}[1]{\prefix{\aca}{#1}}

\renewcommand{\eval}[2]{{#2}_{#1}}

\newcommand{\code}[1]{$#1$}

\renewcommand{\eseq}{[]}

\renewcommand{\Update}[3]{{#1}_{[#2 \smallasgn #3]}}

\renewcommand{\ttdef}{\mathbin{:\!:\!=}}

\renewcommand{\cbar}{\mathbin{~|~}}

\newcommand{\mathttbf}[1]{\mbox{\code{#1}}}

\renewcommand{\If}{\mathrel{\keywordfont{if}}}
\renewcommand{\Then}{\mathrel{\keywordfont{then}}}
\renewcommand{\Else}{\mathrel{\keywordfont{else}}}
\renewcommand{\While}{\mathrel{\keywordfont{while}}}
\renewcommand{\Do}{{\mathrel{\keywordfont{do}}}}

\newcommand{\wptrans}[1]{wp(#1)}
\renewcommand{\wpre}[2]{\wptrans{#1}(#2)}
\renewcommand{\wpc}{\wptrans{\cmdc}}

\newcommand{\refstoRO}[1]{\overset{#1}{\refsto}}
\newcommand{\tracesRO}[2]{\overset{#1}{\Meaning{#2}}}

\newcommand{\refstoROm}{\refstoRO{\mm}}
\newcommand{\tracesROm}[1]{\tracesRO{\mm}{#1}}

\renewcommand\notin{\mathbin{\overlay{\hskip0.7pt/}{$\in$}}}

\renewcommand{\asgn}{\mathbin{\mathtt{:\!=}}}

\renewcommand{\RegStoreWord}{{\mathttbf{MOV}}}
\renewcommand{\LoadWord}{{\mathttbf{LDR}}}
\renewcommand{\StoreWord}{{\mathttbf{STR}}}
\renewcommand{\RMWWord}{{\mathttbf{rmw}}}
\newcommand{\GuardWord}{{\mathttbf{BR}}}

\renewcommand{\RegStoreWord}{{\mathttbf{reg}}}
\renewcommand{\LoadWord}{{\mathttbf{load}}}
\renewcommand{\StoreWord}{{\mathttbf{store}}}
\renewcommand{\RMWWord}{{\mathttbf{rmw}}}
\renewcommand{\GuardWord}{{\mathttbf{guard}}}
\renewcommand{\LoadWord}{{\tt LDR}}
\renewcommand{\StoreWord}{{\tt {STR}}}

\renewcommand{\regstore}[2]{\mathtt{\RegStoreWord~#1 , #2}}
\renewcommand{\load}[2]{\mathtt{\LoadWord~#1 , #2}}
\renewcommand{\rmw}[3]{\mathtt{\RMWWord~#1,#2,#3}}

\def\keywordfont{\mathbf}

\renewcommand{\ffence}{\keywordfont{fence}}
\renewcommand{\cfence}{\fencep{\mathsf{ctrl}}}
\renewcommand{\wwfence}{\fencep{\mathsf{ww}}}
\newcommand{\llfence}{\fencep{\mathsf{ll}}}

\newcommand{\localss}[1]{\locals{\sigma}{#1}}
\newcommand{\localssc}{\localss{\cmdc}}

\renewcommand{\prefix}[2]{#1 \cbef #2}
\renewcommand{\prefixa}[1]{\prefix{\aca}{#1}}

\renewcommand{\vnfi}[2]{#1 \notin \fv{#2}}
\renewcommand{\loadDistinct}[2]{\mx{\sv{#1}}{\sv{#2}}}

\begin{document}

\title{Parallelized sequential composition, pipelines, and hardware weak memory models}

\author{Robert J. Colvin}

\affiliation{
	\institution{
	Defence Science and Technology Group, Australia 
	} 
	\department{
	The University of Queensland
	}
}
\email{r.colvin@uq.edu.au}

\OMIT{
\author{Graeme Smith}

\affiliation{
	\institution{
	Defence Science and Technology Group, Australia
	}
	\department{
	The University of Queensland
	}
}
}

\begin{abstract}

Since the introduction of the CDC 6600 in 1965 and its ``scoreboarding'' technique processors have not (necessarily) executed instructions in program order.
Programmers of high-level code may sequence independent instructions in arbitrary order, and it is a matter of significant programming abstraction and
computational efficiency that the processor can be relied upon to make sensible parallelizations/reorderings of low-level instructions to take advantage of, for
instance, multiple arithmetic units.  At the architectural level such reordering is typically implemented via a per-processor pipeline, into which instructions
are fetched in order, but possibly committed out of order depending on local considerations, provided any reordering preserves sequential semantics from that
processor's perspective.  However, multiprocessing and multicore architectures, where several pipelines run in parallel, can expose these processor-level
reorderings as unexpected, or ``weak'', behaviours.  Such weak behaviours are hard to reason about, and (via speculative execution) underlie at least one class
of widespread security vulnerability.  

In this paper we introduce a novel program operator, \emph{parallelized sequential composition}, which can be instantiated with a function $\mm$ that controls
the reordering of atomic instructions.  The operator generalises both standard sequential composition and parallel composition, and when appropriately
instantiated exhibits many of the weak behaviours of the well-known hardware weak memory models TSO, Release Consistency, ARM, and RISC-V.  We show that the use
of this program-level operator is equivalent to sequential execution with reordering via a pipeline.  Encoding reordering as a programming language operator
admits the application of established compositional techniques (such as Owicki-Gries) for reasoning about weak behaviours, and is convenient for abstractly
expressing properties from the literature such as sequential consistency.  The semantics and theory is encoded and verified in the Isabelle/HOL theorem prover,
and we give an implementation of the pipeline semantics in the Maude rewriting engine and use it empirically to show conformance of the semantics against
established models of ARM and RISC-V, and elucidate some stereotypical weak behaviours from the literature.

\end{abstract}

\ccsdesc[500]{Software and its engineering~Semantics}
\ccsdesc[300]{Software and its engineering~Concurrent programming languages}
\ccsdesc[100]{Computing methodologies~Computer algebra systems}

\keywords{Semantics, pipelines, weak memory models}




\maketitle

\section{Introduction}

The 1960s saw significant improvements in processor efficiency, including allowing out-of-order instruction execution in cases
where program semantics would not be lost (\eg, the ``scoreboarding'' technique of the CDC 6600 \cite{CDC6600}) and maximising use of multiple
computation units (\eg,
Tomasulo's algorithm \cite{Tomasulo67},
implemented in the IBM System360/91 \cite{IBM360/91}).
These advances meant that instructions could be
distributed in parallel among several subunits and their results combined, provided the computation of one did not depend on an incomplete result of another.
Furthermore,
interactions with main memory can be relatively slow in comparison to local computation,
and so allowing independent loads and stores to proceed in parallel also improved throughput.
Even more complex is speculative execution, in the sense of
guessing the result of a branch condition evaluation and transiently executing down that path.
These features had the effect of greatly increasing the speed of processors, and without any visible intrusion on
programmers: the conditions under which parallelization could take place ensured the \emph{sequential semantics} of any computation was
maintained.  

When concurrency is used, either explicitly as threads sharing a single processor or via multicore architectures,
the effect of out-of-order execution may be exposed, as the reordering of 
accesses of shared memory
can dramatically change concurrent behaviour.
This has provided a challenge for developing efficient, correct and secure concurrent software for modern processors that
communicate via shared memory \cite{SharedMemModelsTutorial,Hill98,Batty2015}.
Order can be restored by injecting artificial dependencies between instructions (usually called fences or barriers), but the performance cost is significant;
for instance, performance concerns hamper the 
widespread mitigation of the Spectre class of security vulnerabilities, despite their seriousness \cite{Spectre,SpectrePrime}.

To reason about the impact of reorderings on program behaviour we introduce \emph{parallelized sequential composition} as a
primitive operator of an imperative language.  The program `$c_1 \ppseqm c_2$', for some function on instructions $\mm$, may execute $c_1$ and $c_2$ in order, or it may
interleave actions of $c_2$ with those of $c_1$ provided $\mm$ allows it.  We give the weakest $\mm$ such that $c_1 \ppseqm c_2$ preserves the intended
sequential semantics of a single process despite executing some instructions out of order, and show how modern \emph{memory models} are strengthenings of this concept, with the
possible addition of further instruction types to restore order.  Based on restrictions built into processors as early as the 1960s, we show that this concept of
thread-local reorderings explains many of the behaviours observed on today's multicore architectures, such as TSO \cite{Intel64ArchManual} and ARM \cite{ARMv8A-Manual}, and
conforms to RISC-V \cite{RISC-V-ISAManual2017}.  We derive language-level properties of \pseqc, and use these to explain possibly unexpected behaviours algebraically.  In
particular, under some circumstances the possible reorderings can be reduced to a nondeterministic choice between sequential behaviours, and then existing techniques for
reasoning about concurrency, such as the Owicki-Gries method \cite{OwickiGries76}, can be employed directly.  The language, its semantics, and the properties in the paper
are encoded and machine-checked in the Isabelle/HOL theorem prover \cite{IsabelleHOL,Paulson:94} (see supplementary material).  


\OMIT{
For instance, a classic example of weak memory models producing unexpected behaviour is the ``store buffer'' pattern below \cite{LitmusTests}.
Assume that all variables are initially 0, that $r_1$ and $r_2$ are thread-local variables, and that $x$ and $y$ are shared variables.
\begin{equation}
	(x \asgn 1 \cbef r_1 \asgn y)
	\pl 
	(y \asgn 1 \cbef r_2 \asgn x)
\end{equation}
It is possible to reach a final state in which $r_1 = r_2 = 0$ in several weak memory models:  the two assignments in each
process are independent (they reference different variables), and hence can be reordered.  From a sequential semantics perspective,
reordering the assignments in process 1, for example, preserves the final values
for $r_1$ and $x$.
}



In \refsect{overview} we provide a foundation for instruction reordering, and provide a range
of theoretical memory models.
In \refsect{pipeline} we give a straightforward abstract semantics for a hardware pipeline with reordering.
In \refsect{syntax} we fully define the syntax and semantics of an imperative language, \impro, that includes conditionals, loops, and \pseqc. 
In \refsect{eff-semantics}
we show how standard Hoare logic and weakest preconditions can be used to reason about \impro.
We then define TSO (\refsect{tso}), Release Consistency (\refsect{rc-model}), ARM (\refsect{armv8}) and RISC-V (\refsect{riscv}) as instances of \pseqc,
showing conformance via both general properties and empirically (ARM and RISC-V).
We discuss related work in \refsect{related-work}.

\section{Foundations of instruction reordering}
\labelsect{overview}
\labelsect{overview-reordering}
\labelsect{foundations}

In this section we describe instruction reordering in a simple language containing a minimal set of instruction types and 
the novel operator \pseqc (we give a richer language in \refsect{syntax}).  Instances of this operator are parameterized by a function on instructions,
which for convenience we call a \emph{memory model}.%
\footnote{
Memory models may embrace global features in addition to weak behaviours due to pipelining; since many behaviours of weak memory models are explained 
by pipeline reordering we use this more general term.
}
We use this foundation to explore theoretically significant memory models that underlie modern processors.

We start with a basic imperative language containing just assignments and guards as actions (type $\Instr$) with parallelized sequential
composition as combinator.
\begin{equation*}
	\aca \sspace \ttdef \sspace x \asgn e \csep \guarde 
	\qquad \qquad
	c \sspace \ttdef \sspace \Skip \csep \aca \csep c_1 \ppseqm c_2
\end{equation*}
An assignment $x \asgn e$ is the typical notion of an update, encompassing
stores and loads, while a guard action $\guarde$ represents a test on the state (does expression $e$ evaluate to $\True$) which can be used to
model conditionals/branches.
A command is either the terminated command $\Skip$ (corresponding to a \T{no-op}), an action $\aca$, or the composition of two commands
according to some memory model $\mm$.

The intention is that a simple command $\aca \ppseqm \acb$ is free to execute instruction $\acb$ before instruction $\aca$ (possibly with modifications due to forwarding, 
described below) provided the constraints of $\mm$ are
obeyed.  
Clearly this potential reordering of instructions may destroy the programmer's intention if unconstrained; however the
reordering (or parallelization) of \emph{independent} instructions can potentially be more efficient than the possibly arbitrary order
specified by the programmer.
For example, consider a load followed by an update $r \asgn x \ppseqm y \asgn 1$.  If $x$ is a shared variable then retrieving its value from
main memory may take many processor cycles.  Rather than idle the independent instruction $y \asgn 1$ can be immediately issued without
compromising the programmer's intention \emph{assuming a single-threaded context}.
In general, two assignments can be reordered if they preserve the sequential semantics on a single thread.

A memory model $\mm$ is of type $\Instr \fun \power (\Instr \cross \Instr)$.
However we will typically express a memory model as a binary relation on instructions, the reordering relation, with the 
implicit application of a forwarding function, which is
either a straightforward variable substitution or the identity function.  
We first consider
the reordering relation and then forwarding, which
significantly complicates matters.

\paragraph{Reorderings}
First consider the base of a memory model, the ``reordering relation'', which is a binary relation on instructions.
We write $\aca \rom \acb$ if 
$\acb$ may be reordered before instruction $\aca$ according to the reordering relation of memory
model $\mm$, and $\aca \nrom \acb$ otherwise.
We define the special case of the \emph{sequentially consistent} memory model as that where no reordering is possible.
\begin{mmdef}[$\SCmm$]
\labelmm{sc}
For all $\aca,\acb \in \Instr$,
$
	\aca \nroSC \acb 
$.
\end{mmdef}
We may now explicitly state the notion of a \emph{sequential} model \cite{LamportSC}, which is the minimal property that any practical memory model 
$\mm$ should establish.  Letting the `effect' function $\effword$ return the relation between pre- and post-states for a program 
(see \refsect{eff-semantics})
we can state this property as follows.
\begin{definition}
\labeldefn{seqmm}
\labeleqn{eff-pseqc}
$\mm$ is \emph{sequential} if
$
	\eff{c_1 \ppseqm c_2} \subseteq \eff{c_1 \ppseqs c_2}
$.
\end{definition}
That is, $\mm$ is sequential if its reorderings give the same results (on a single thread) as when executed
in program order.
The weakest sequential memory model is one that allows reordering exactly when sequential semantics is maintained.
\begin{mmdef}[$\Effmm$]
\labelmm{wpmm}
$
	\aca \roEff \acb \sspace iff \sspace
		\eff{\acb \ppseqs \aca} \subseteq \eff{\aca \ppseqs \acb} 
$
\end{mmdef}
\begin{theorem}%
\labelth{eff-is-seq}
$\Effmm$ is \emph{sequential}.
\end{theorem}
\begin{proof}
The property $\aca \roEff \acb$ lifts to commands.
\end{proof}

$\Effmm$ is impractical since processors cannot 
make semantic judgements about the overall effect of two instructions dynamically;
however there are some simple syntactic constraints which guarantee sequential semantics.
As such we propose the following memory model as the weakest possible that is of practical use.
\begin{mmdef}[$\Gzmm$]
\labelmm{G0}
$
	\aca \roGz \acb ~iff~  \mbox{$\mx{\wv{\aca}}{\fv{\acb}}$ ~ and ~ $\mx{\wv{\acb}}{\fv{\aca}}$ }
$
\end{mmdef}
\refMM{G0} (abbreviated \refmm{G0}) allows instructions to be reordered based on a simple \emph{syntactic} test, namely, is any variable that $\acb$ references ($\fv{\acb}$)
modified by $\aca$,
and vice versa, using naming conventions below.
\begin{align}
	\fve, \fva
	&\sspace
	\mbox{Free variables in expr. $e$ or action $\aca$}
	\labeleqn{defn-fve}
	\\
	\wv{\aca}, \rva &\sspace \mbox{Write, read variables}
	\labeleqn{defn-wv}
	\\
	\mx{s_1}{s_2} \sdef &\sspace 
		\mbox{$s_1 \cap s_2 = \ess$ (mutual exclusion)} 
	\labeleqn{defn-vars-mutex}
\end{align}
Reorderings eliminated by \refmm{G0} 
include 
$
		(x \asgn 1 \nroGz x \asgn 2)
$, 
$
		(x \asgn 1 \nroGz r \asgn x)
$
and
$
		(r \asgn x \nroGz x \asgn 1)
$.
\OMIT{
For instance, 
simplifying the constraint of \refmm{G0} for two assignments,
$x \asgn e \roGz y \asgn f$, gives
\begin{equation}
\labeleqn{reordering-principle}
	\mbox{%
		(i) $x \neq y$
		\ \ (ii) $x \notin \fv{f}$
		\ \ and \ \ (iii) $y \notin \fv{e}$}
\end{equation}
}
In general, if $\wva \subseteq \rvb$, for $\aca, \acb \in \Instr$, then there is a \emph{data dependency} between $\aca$ and $\acb$,
that is, the value of a variable that $\acb$ depends on is being computed by $\aca$.
It is straightforward that
$
	(\aca \roGz \acb)
	\imp
	(\aca \roEff \acb)
$ \ie,
$\Gzmm$ is stronger than $\Effmm$.
We can therefore infer that $\Gzmm$ is \emph{sequential}.
\begin{theorem}
If 
$ \mm$ is stronger than $ \Effmm $
then $\mm$ is \emph{sequential}.
\end{theorem}
\begin{proof}
A stronger model admits fewer behaviours
\refrule{pseqcB}.
\end{proof}

\OMIT{
On a single processor register operations can be parallelized by utilising multiple ALUs.  In a program like
$
	r1 \asgn r2 + r3 \ppseqm r4 \asgn r5 \times r6
	$
it has been the case since the 1950s that these are executed in parallel where possible; the order given is simply the whim of the programmer and could
easily have been given in the opposite order, and hence can also be parallelized.  Registers operations do not need to interact with main memory and
this is call instruction-level parallelism; it is a special case of memory models.
}

The memory model $\Gzmm$ is lacking in the age of multicore processors because 
it does not require two consecutive loads of the same \emph{shared} variable to be performed in order. For instance, consider the following 
concurrent processes. 
\begin{equation}
\labeleqn{2loads}
	r_1 \asgn x \ppseqGz r_2 \asgn x
	\quad
	\pl
	\quad
	x \asgn 1 \ppseqGz x \asgn 2
\end{equation}
The two loads should read values of $x$ in a globally ``coherent'' manner, that is, the value
for $x$ loaded into $r_1$
must occur no later than that loaded by $r_2$.
Hence program \refeqn{2loads} should not reach the
final state $r_1 = 2 \land r_2 = 1$.
However, although $x \asgn 1 \nroGz x \asgn 2$, we have $r_1 \asgn x \roGz r_2 \asgn x$ in the first thread.

To cope with shared variables and coherence we divide the set of variables, $\AllVars$, into mutually exclusive and exhaustive sets $\SharedVars$ and $\LocalVars$,
with specialised definitions below.
\begin{align}
	\sva, \rsva, & \wsva
	\sspace \mbox{As \refeqn{defn-fve}, \refeqn{defn-wv}, restricted to $\SharedVars$}
	\labeleqn{defn-sve}
	\\
	\OMIT{
	\fva \int \SharedVars,
	\sspace
	\rva \int \sva,
	\sspace
	\wva \int \sva,
	\sspace
	\mbox{respectively}
	\\
	}
	\isStorea \sdef &\sspace \wsva \neq \ess \land \rsva = \ess 
	\labeleqn{isStore}
	\\
	\isLoada \sdef &\sspace \wsva = \ess \land \rsva \neq \ess 
	\labeleqn{isLoad}
	\\
	\isRega~\mbox{iff} &\sspace \mbox{$\aca$ is an assign. or guard and $\sv{\aca} = \ess$}
	\labeleqn{wsv} 
\end{align}
To maintain ``coherence per location'' we extend $\Gzmm$ to $\Gmm$ by adding a constraint on the loaded shared variables.
Additionally, since we are now explicitly concerned with concurrent behaviour, we add
a fence instruction type to restore order, \ie,
$
	\aca \ttdef \ldots \csep \ffence
$.
We call this an ``artificial'' constraint, since it is not based on ``natural''
constraints arising from the preservation of sequential semantics.
\begin{mmdef}[$\Gmm$]
\labelmm{G}
$
	\aca \roG \acb \ \ \ltiff{\ $\aca \roM{\Gzmm} \acb \land \mx{\rsv{\aca}}{\rsv{\acb}}$}
$,
except \\
$
	\ 
	\qquad
	\sspace
	\aca \nroG \ffence \nroG \alpha  
$.
\end{mmdef}
When specifying a memory model we typically give the base relation first, and then list the ``exceptions'', which take precedence.
\refmm{G} strengthens the condition of \refmm{G0} to require loads from main memory to be kept in program order per shared variable.
In addition fences block reordering, reinstating program-order execution explicitly if desired by the programmer (at the potential cost 
of efficiency).
We let $\aca \rom \acb \rom \acc$ abbreviate $\aca \rom \acb \land \acb \rom \acc$, and similarly for $\nrom$.

We consider $\Gmm$ to be the weakest memory model of practical consideration in a concurrent system as it
maintains both coherence-per-location and sequential semantics.
\begin{definition}
Model $\mm$ is \emph{coherent} if it is stronger than $\Gmm$.
\end{definition}
Most modern processors are coherent.
A memory model that is not coherent, and is the obvious weakest dual of \refmm{sc},
is one that allows any instructions to be reordered under any circumstances.  If we disallow forwarding in this model (discussed in the next section),
this weakest memory model
corresponds to parallel composition.
\begin{mmdef}[$\PARmm$]
\labelmm{PARmm}
For all $\aca, \acb \in \Instr$,
$
	\aca \roPAR \acb 
$
\end{mmdef}
We may now define
$
	c \pl d \sdef c \ppseq{\PARmm} d
$,
lifting instruction-level parallelism to thread-level parallelism.

\newcommand{\scprop}{sequential\xspace consistency\xspace}


The key point about the $\SCmm$ memory model (\refmm{sc}) is that reasoning is ``straightforward'', or classical, in that all the accepted techniques work.
This is the property of \emph{\scprop} \cite{LamportSC}, formalised below.
\begin{definition}
\labeldefn{undermm}
Command $\cunderm$ is structurally identical to $\cmdc$ but every \pseqc,
except for instances of $\PARmm$,
is paramet\-erized by $\mm$.
\end{definition}
\begin{definition}[Sequentially consistent]
\labeldefn{seqcst}
A memory model $\mm$ is sequentially consistent if, for any programs $\cmdc$ and $\cmdd$, ignoring local variables,
$
	\undermm{\cmdc}{\mm} \pl \undermm{\cmdd}{\mm} ~~~\refeq ~~~
	\undermm{\cmdc}{\SCmm} \pl \undermm{\cmdd}{\SCmm} 
$.
\end{definition}
\OMIT{
The notation $\cunderm$ lets us consider programs executed under some specific model $\mm$
(parallel composition is left unchanged as we are typically only interested in
thread-local effects).
In our framework the notion of equivalence is based on sets of behaviours (traces; see \refsect{trace-semantics}), 
in which local variables can be hidden in the usual way.
}
By definition $\SCmm$ is sequentially consistent, however even 
sequentially consistent uniprocessors are not as strong as \refmm{sc}, for instance,  some speculate loads and reissue them if a change is detected 
(relatively quickly via the cache \cite{MIPSR10000}). 
Note the difference between \emph{sequential} and \emph{sequentially~consistent}: a sequentially consistent memory model is sequential, but not vice versa.  
None of TSO, ARM or RISC-V are sequentially consistent in general, but are for programs in a particular form, \eg,
where shared variables are accessed according to a lock-based programming discipline.

\OMIT{
None of TSO, ARM or RISC-V are sequentially consistent.  A key property of programming language-level memory models is that sequential consistency holds for
programs of a particular form, e.g., where shared variables are accessed according to a lock-based programming discipline.
In this paper we reason about a program $c$ executed under memory model $\mm$ using three approaches: show that the constraints of $\mm$ mean that $\cunderm \refeq
\cundersc$, and so existing techniques for reasoning about concurrent sequential programs can be applied; enumerate the possible reorderings of $\cunderm$ so that it
can be approached as a (nondeterministic) choice between concurrent sequential programs, and again standard techniques can be applied; or find a particular ordering
(behaviour) $c'$ of $\cunderm$
that contradicts some desired property $P$ of $\cunderm$, and from this infer that $P$ does not hold for $\cunderm$.
}

\OMIT{
This is given by the following weaker memory model $\POwmm$.
\begin{mmdef}{$\POwmm$}
\labelmm{po-wmm}
\begin{eqnarray}
	\aca \nroPOw &\fence& \nroPOw \aca
	\ltif{$\sv{\aca} \neq \ess$}
	\labeleqn{powmm-fence}
	\\
	\aca \roPOw \acb &iff& 
		\mbox{$\aca \roG \acb$ \sspace and\sspace  $\sva = \ess \lor \svb = \ess$}
	\labeleqn{powmm-gen}
\end{eqnarray}
\end{mmdef}
In $\POwmm$ fences only block instructions that reference \SharedVars variables \refeqn{powmm-fence},
and other reorderings are allowed provided it is not the case that both reference \Svars (and the usual dependencies hold) \refeqn{powmm-gen}.
This memory model allows local optimisations and even speculative execution, provided
none are dependent on (unresolved) interactions with the global state.  
\begin{theorem}
\POwmm is sequentially consistent for \SharedVars-only processes.
\end{theorem}
\begin{proof}
The program order of instructions involving shared variables is maintained, and the reorderings allowed on locals-only instructions are permitted
only if sequential behaviour is preserved locally.
\end{proof}
}

\paragraph{Forwarding}
\labelsect{forwarding}

We now complicate matters significantly by considering \emph{forwarding}, where the effect of an earlier operation 
can be taken into account when deciding if instructions can be reordered.%
\footnote{We use the term ``forwarding" from ARM and POWER \cite{HerdingCats}, sometimes referred to as ``bypassing'' in TSO \cite{x86-TSO}.}    
For instance, given a program
$
	x \asgn 1 \ppseqG r \asgn x
$,
we have $x \asgn 1 \nroG r \asgn x$ because $\wv{x \asgn 1} = \{x\} \subseteq \fv{r \asgn x}$,
violating \refmm{G0}.
In practice however it is possible to \emph{forward} the new value of $x$ to the later instruction -- it is clear that the value assigned to $r$ will be 1 if $x$ is local, and in
any case is a valid possible assignment to $r$ even if $x$ is shared.
We define 
$\fwd{\alpha}{\beta}$, representing the effect of forwarding the (assignment) instruction $\alpha$ to $\beta$, 
where the expression $\Repl{f}{x}{e}$ is $f$ with references to $x$ replaced by $e$.
\begin{definition}[Forwarding]
\labeldefn{forwarding}
	$\fwd{\aca}{ \acb} = \acb$, except
\begin{gather*}
	\fwd{x \asgnlbl e}{(y \asgn f)} = y \asgn (\Repl{f}{x}{e})
		\qquad
	\fwd{x \asgnlbl e}{\guardf} = \guard{\Repl{f}{x}{e}}
		\qquad
	\labeleqn{fwd-g}
\end{gather*}
\end{definition}

\OMIT{
\begin{equation}
\labeleqn{assignment-forwarding}
	\fwd{x \asgnsmall e}{(y \asgn f)}
	~~=~~
	y \asgn (\Repl{f}{x}{e})
\end{equation}
}

Forwarding and reordering are combined to form
a memory model as follows, where 
the effect of forwarding is taken into account \emph{before} calculating reordering.
\begin{gather}
	\mm 
	\sspace
	\sdef 
	\sspace
	\lambda \aca. 
		\{(\fwd{\aca}{\acb} , \acb) | 
			\aca \rom \fwdab 
		\}
\labeleqn{mm-def}
\\
\labeleqn{defn-roa}
\robpab
\sspace
\sdef
\sspace
(\acb', \acb) \in \mm(\aca)
\end{gather}
Thus a memory model $\mm$ for a given action $\aca$ returns a set of pairs $(\acb', \acb)$ where
$\acb$ reorders with $\aca$, after the effect of forwarding $\aca$ to $\acb$ ($\acb'$) is taken into account.
For convenience we sometimes use the notation $\robpab$ which notationally conveys the bringing forward of $\acb$ with respect 
$\aca$.
For example, since
$\fwd{x \asgnsmall 1}{(r \asgn x)} = (r \asgn (\Repl{x}{x}{1})) = r \asgn 1$, 
the load $r \asgn x$ ``reorders'' with $x \asgn 1$, becoming $r \asgn 1$.
\begin{equation}
\labeleqn{rotrip-example}
	\rocmd{r \asgn 1}{x \asgn 1}{r \asgn x}{\Gzmm}
\end{equation}
Forwarding is significant because it can change the orderings allowed in a non-standard way, since
a later instruction that was blocked by $r \asgn x$ may no longer be blocked,
and potentially can now also be reordered before $x \asgn 1$.
Of course, this is potentially a significant efficiency gain, because local computation can proceed using the value 1 for $x$ without waiting for the update
to propagate to the rest of the system.

Memory models with out-of-order execution typically use forwarding.
(As noted above, one exception is $\PARmm$ for interleaving parallel).
In $\Gmm$ and $\Gzmm$ reordering is symmetric, however when calculated after the effect of forwarding is applied there are instructions 
that may be reordered in one direction but not the other.  In general a reordering relation is neither reflexive nor transitive.

\OMIT{
We can make the following statement about reordering in the $\Gzmm$ model taking forwarding into account, stated as a triple from
\refeqn{defn-roa}.
Part (ii) is a weakening of the equivalent from \refeqn{reordering-principle}, 
being an expansion of $x \notin \fv{\Repl{f}{x}{e}}$.  
\begin{equation*}
\labeleqn{reordering-principle-with-forwarding}
    \roact{y \asgn \Repl{f}{x}{e}}{x \asgn e}{y \asgn f}{\Gzmm}
    \quad
    \mbox{~~iff~~}
    \quad
    \mbox{%
        (i) $x \neq y$
        \ \ (ii) $x \in \fvf \imp x \notin \fve$
        and \ \ (iii) $y \notin \fv{e}$}
\end{equation*}
}


\OMIT{
For $\Gmm$ we get
\begin{equation}
    \roact{y \asgn \Repl{f}{x}{e}}{x \asgn e}{y \asgn f}{\Gmm}
    \mbox{~~iff~~}
    \mbox{
		\begin{tabular}{l}
        (i) $x \neq y$
        (ii) $x \in \fvf \imp x \notin \fve$
        (iii) $y \notin \fv{e}$
		\\ 
        (iv) $x \in \fvf \imp \sv{e} = \ess$
		and
		\\
        (v) $x \notin \fvf \imp \mx{\sv{e}}{\sv{f}}$
		\end{tabular}
	}
\end{equation}
}

\paragraph{Memory model refinement}

We define model refinement as a strengthening per action.
\begin{definition}
	\labeldefn{roref}
$
	\mm_1 \roref \mm_2
	\sdef 
	\forall \aca @ \mm_2(\aca) \subseteq \mm_1(\aca)
	$
\end{definition}
\OMIT{
Hence, 
\[
	\mm_1 \roref \mm_2
	\qquad iff \qquad
	\forall \acb', \aca, \acb @
	\roact{\acb'}{\aca}{\acb}{\mm_2} \imp
	\roact{\acb'}{\aca}{\acb}{\mm_1} 
\]
}
As we explore in the rest of the paper,
the Total Store Order memory model ($\TSOmm$) strengthens $\Gmm$ considerably (or alternatively, weakens 
$\SCmm$ for the particular case of stores and loads),
while $\ARMmm$ strengthens $\Gmm$ to prevent stores from coming before branches.
$\ARMmm$, $\RVmm$, and the release consistency models $\RCmm$ and $\RCSCmm$ are related as below, focusing on their common instruction types;
since each introduces unique instruction types a direct comparison is not possible.
\begin{theorem}[Hierarchy of models]~\\
$
	\Effmm \roref \Gzmm \roref \Gmm \roref \RCmm \roref \RVmm \roref \RCSCmm \roref \ARMmm \roref \TSOmm \roref \SCmm
$.
\end{theorem}
\begin{proof}
	Straightforward from definitions.
\end{proof}
This result is similar to other classifications \cite{HierarchyWMM,FrigoLuchangco98,ReleaseConsistency90}.
Note that $\PARmm$ does not fit into the hierarchy because it does not allow forwarding.

\paragraph{Well-behaved models}
A memory model could theoretically allow arbitrary reorderings and effects of forwarding;
however from a reasoning perspective we make the following definition of well-behaved
memory models.
\begin{definition}
\labeldefn{well-behaved}
A memory model $\mm$ is \emph{well-behaved} if it is $\SCmm$ or if:
\OMIT{
\begin{equation*}
	i)~
		\roactm{\acb'_1}{\aca}{\acb} \land 
		\roactm{\acb'_2}{\aca}{\acb} 
		\entails
		\acb'_1 = \acb'_2
		\qquad
	ii)~
		\roactm{\acb'}{\aca}{\acb} 
		\entails
		\acb' = \fwdab \lor \acb' = \acb
		\qquad
	iii)~
		\aca \rom \tau 
		\labeleqn{mm-cst}
\end{equation*}
}
i) the result of reordering is deterministic; ii) any resulting action $\acb'$ must 
arise from the application of the forwarding function (\refdefn{forwarding}), or have no change at all, \ie, we do not allow arbitrary effects of forwarding;
and 
iii) if an action allows any reordering then it must allow reordering with \emph{internal} (``silent'') steps.
\end{definition}
Conditions i) and ii) ensure determinacy and sequential semantics, while condition iii) simplifies reasoning
(silent steps are defined in
in \refsect{semantics}).
All the memory models we consider are well-behaved, with the exception of \PARmm.

\OMIT{
Combining ii) and iii) we have $\roact{\tau}{\aca}{\tau}{\mm}$, from which we can infer
for any command $c$, $\rocmdm{\tau}{c}{\tau}$.
Property iii) says that if a processor allows, for instance, a load to be executed early, then it must also allow silent (internal) steps of the processor
to be executed early; this is required due to our weak semantics.
}

\OMIT{
A consequence of iii) is that in a
program $c_1 \ppseqm c_2$, if $c_2$ diverges, \eg, is an infinite loop with no visible steps, then the composition may diverge.  
We take this as a natural negative side effect of allowing computation to proceed in parallel, that is, theoretically, unless measures are taken explicitly by a
processor to prevent it, future simplifications could run ahead indefinitely and prevent progress.  In real systems the pipeline is finite and infinite speculative
computation is unlikely to happen.
}

\section{Pipeline semantics}
\labelsect{pipeline}

\newcommand{\pipelineword}{\mathbf{pline}}
\newcommand{\plt}{\T{t}}
\newcommand{\plcat}{\cat }
\newcommand{\pipeline}[3]{(\pipelineword_{#1}~#2 :~ #3)}
\newcommand{\pline}[3]{\pipeline{#1}{#2}{#3}}
\newcommand{\plinem}[2]{\pline{\mm}{#1}{#2}}
\newcommand{\plinemp}[1]{\plinem{\plt}{#1}}
\newcommand{\plinempc}{\plinemp{\cmdc}}
\newcommand{\plinemepc}{\plinem{\epl}{\cmdc}}

\newcommand{\epl}{\eseq}

Before defining a full language we consider how a reordering relation can be used to define
the semantics of a processor pipeline where instructions can be reordered.
The command $\plinempc$ represents the execution of $\cmdc$ with $\plt$ a sequence of instructions that are currently (in) the pipeline;
instructions in $\plt$ can be issued to the wider system in order, or out-of-order according to $\mm$.
Assume that $\cmdc$ is executed sequentially, (\ie, is of the form $\cundersc$),
then the behaviour of the pipeline is
as follows.

\ruledefNamed{0.95\columnwidth}{Pipeline - fetch}
{pline-fetch}{
	\Rule{
		c \tra{\aca} c'
	}{
		\plinempc
		\tra{\tau} 
		\plinem{\plt \cat \aca}{c'}
	}
}

\ruledefNamed{0.95\columnwidth}{Pipeline - commit}
{pline-commit}{
	\Rule{
		\rocmdm{\aca'}{\plt_1}{\aca}
	}{
		\plinem{\plt_1 \cat \aca \cat \plt_2}{c}
		\tra{\aca'} 
		\plinem{\plt_1 \cat \plt_2}{c}
	}
}

\OMIT{
\ruledefNamed{0.48\columnwidth}{Pipeline - commit}
{pline-commit}{
		\plinem{\aca \cat \plt}{c}
		\tra{\aca} 
		\plinempc
}
\ruledefNamed{0.46\columnwidth}{Pipeline - propagate}
{pline-propagate}{
	\Rule{
		c \tra{\aca} c'
		\qquad
		\rocmdm{\aca'}{\plt}{\aca}
	}{
		\plinempc
		\tra{\aca'}
		\plinemp{c'}
	}
}
}


Rule \refrule{pline-fetch} states that the next instruction in $\cmdc$ in sequential order can be ``fetched'' and placed at the end of the pipeline;
$c$ is effectively the ``code base''.  
We use `$\cat$' for appending lists, and for convenience allow it to apply to individual elements as well.
The promoted step $\tau$ represents an internal action of the processor, which is ignored by the rest of the system.
Rule \refrule{pline-commit} states that an instruction $\aca$ somewhere in the pipeline can be ``committed'' out of order, provided it can be reordered
with all prior instructions currently in the pipeline. The notation  
$\rocmdm{\aca'}{\plt}{\aca}$ (cf. \refeqn{defn-roa}) is a shorthand for $(\aca', \aca) \in \mm(\plt)$, \ie, lifting $\mm$ from a function on
instructions to sequences of instructions,
that is,
$
	\mm(\eseq) = \id 
$
and 
$
	\mm(\aca \cat \plt) = \mm(\aca) \comp \mm(\plt)
$,
where `$\eseq$' is the empty list, `$\id$' is the identity relation, and `$\comp$' is relational composition.

Consider executing the program $r_1 \asgn x \scomp r_2 \asgn y \scomp \cmdc$ in a pipeline.  
Both loads can be fetched (in order) into the pipeline by \refrule{pline-fetch}, 
but then, assuming 
$
r_1 \asgn x \rom r_2 \asgn y 
$,
$r_2 \asgn y$ may be committed before $r_1 \asgn x$ by \refrule{pline-commit} (or further instructions from $\cmdc$ could be fetched
and potentially committed).
The fetch/commit rules abstract from other stages in a typical pipeline (see, e.g., \cite{PipeCheck}), for instance, the two loads above would be \emph{issued} to the system in order,
with the out-of-order commit corresponding to the second load being serviced earlier by the memory system.

We encoded this semantics straightforwardly in the Maude rewriting engine \cite{Maude,SOSMaude} as a model checker (see supplementary material).
We find this processor-level view to be convenient for showing equivalence with other models of reordering, for instance, the store buffer model of TSO
(\refsect{tso-buf})
and an axiomatic specification of ARM (\refsect{arm:axiomatic}).
However it is not directly amenable to established techniques for reasoning about programs such as 
Hoare Logic \cite{Hoare69}, Owicki-Gries (OG) \cite{OwickiGries76} or rely/guarantee \cite{Jones-RG1},
which are over the structure of a program.
As such we now consider embedding reordering into a typical imperative language, where the following theorem holds
(recall \refdefn{undermm}).
\begin{theorem}
	\labelth{pl=pseqc}
For well-behaved $\mm$,
$
	\cunderm
	~=~
	\plinem{\epl}{\cundersc}
$
\end{theorem}
\begin{proof}
By induction on traces.
\end{proof}

\section{An imperative language with \pseqc}
\labelsect{syntax}

In this section we give the syntax and semantics for an imperative programming language, ``\impro'', 
which uses \pseqc (extending that of \refsect{overview-reordering}).

\paragraph{Syntax}

The syntax of \impro is given below.
\begin{align}
\omit\rlap{
$
	\aca 
	\sspace\ttdef\sspace
		x \asgn e
		\cbar
		\guarde
		\cbar
		\fencepf
		$
}
\notag
\\
\omit\rlap{
$
	\cmdc 
	\sspace\ttdef\sspace
		\Skip 
		\cbar
		\aca
		\cbar
		\cmdc_1 \ppseqm \cmdc_2
		\cbar
		\cmdc_1 \choice \cmdc_2
		\cbar
		\iteratecm
		\notag
$
}
\\
	\also
	\cmdc_1 \scomp \cmdc_2 
	\sdef
	\cmdc_1 \ppseq{\SCmm} \cmdc_2 
	\sspace&\quad\sspace
	\cmdc_1 \pl \cmdc_2 
	\sdef
	\cmdc_1 \ppseq{\PARmm} \cmdc_2 
	\labeleqn{pl-def}
	\\
	\OMIT{
	\locals{\sigma}{\cmdc} 
	&\sdef&
	\varss{\cmdc} \qquad \mbox{where $\dom(\sigma) = \LocalVars \cap \fv{\cmdc}$}
	\labeleqn{process-structure}
	\\
	\globals{\sigma}{\cmdc} 
	&\sdef&
	\varss{\cmdc} \qquad \mbox{where $\dom(\sigma) = \SharedVars \cap \fv{\cmdc}$}
	\labeleqn{globals-def}
	\\
	}
	\finiteratecm{0} \sdef \Nil
	\sspace&\quad\sspace
	\finiteratecm{n+1} \sdef c \ppseqm \finiteratecmn
	\labeleqn{finiter-def}
	\\
	\also
	(\IFbc)_\mm 
	\sspace&\sdef\sspace
	\guard{b} \ppseqm \cmdc_1 \ \ \choice \ \  \guard{\neg b} \ppseqm \cmdc_2
	\labeleqn{defn-if}
	\\
	(\WHbc)_\mm 
	\sspace&\sdef\sspace
	\iterate{(\guardb \ppseqm \cmdc)}{\mm} \ \ppseqm \ \guard{\neg b}
\labeleqn{defn-while}
\end{align}

The basic actions of a weak memory model are
an update $x \asgn e$, a guard $\guarde$, and a barrier.
(We give an atomic expression evaluation semantics for assignments and guards, which is typically reasonable for assembler-level
instructions.)
The barrier instruction type can be instantiated for some data type which we leave unspecified at this point; particular memory models typically introduce their own
barrier/fence types and we define them in later sections.  We assume that every model has at least a ``full'' fence, and define $\ffence \sdef
\fencep{\fullfencespec}$.
The special
instruction $\tau \sdef \guard{\True}$ is a \emph{silent step} (defined later), having no effect on the state, possibly corresponding to some internal actions of a
microprocessor.

A command $\cmdc$ may be the terminated command $\Skip$, a single instruction,
the parallelized sequential composition of two commands
(parameterised by a memory model), a binary choice, 
or an iteration.
Iterations are parameterised by a memory
model, as they implicitly contain sequencing.

\labelsect{syntax-abbreviations}

In \refeqn{pl-def}
we use \pseqc to define
`$\scomp$' as the usual notion of sequential composition (see \refmm{sc}), 
and 
`$\pl$' as the usual notion of parallel composition (see \refmm{PARmm}).
Finite iteration of a command, $\finiteratecmn$, is the $n$-fold composition of $\cmdc$ with reordering according to $\mm$ \refeqn{finiter-def}.
Conditionals are modelled using guards and choice (where false branches are never executed) \refeqn{defn-if}.
By allowing instructions in $\cmdc_1$ or $\cmdc_2$ to be reordered before the guards one can model \emph{speculative execution}, i.e., early execution of instructions which occur after a branch point \cite{PrimerMemoryConsistency11}.
We define a while loop using iteration \refeqn{defn-while} following \cite{FischerLadner79,KozenKAT}. 
Both conditionals and loops are parameterised by 
a memory model since they include a parallelized sequential composition.

\OMIT{
A well-formed system is structured as the parallel composition of processes within the scope of a shared global state 
$\sigma$, that maps all shared variables to their values.
Each process is structured as a (possibly empty) local state $\sigma$ encompassing command $\cmdc$, i.e., a term
$\locals{\sigma}{\cmdc}$, where $\sigma$ refers only to local variables 
referenced in $\cmdc$ (and $\cmdc$ does not further scope any variables).
\begin{equation}
\labeleqn{system-structure}
	\globals{\sigma}{
		\locals{\sigma_1}{\cmdc_1}
		\pl
		\locals{\sigma_2}{\cmdc_2}
		\pl
		\ldots
	}
\end{equation}
}

\paragraph{Operational semantics}
\labelsect{semantics}

\newcommand{\doms}{\dom(\sigma)}

The meaning of \impro is formalised using an operational semantics, given below.
The operational semantics generates a sequence of syntactic instructions
(as opposed to Plotkin-style pairs-of-states \cite{Plotkin}), allowing syntactic analysis of instructions to decide on reorderings.
We show how to convert straightforwardly to pairs-of-states style in \refsect{eff-semantics}.

\begin{centering}
\def\colwidthA{0.59\columnwidth}
\def\colwidthB{0.38\columnwidth}

\ruledefNamed{\colwidthA}{Action}
{action}{
	\aca \tra{\aca} \Skip
}
\ruledefNamed{\colwidthB}{Iterate}
{iterate}{
	\iteratecm
	\tra{\tau} 
	\finiteratecmn
}
\ruledefNamed{\colwidthA}{Choice}
{nondetL}{
	\cmdc \choice \cmdd \tra{\tau} \cmdc
}
\ruledefNamed{\colwidthB}{Choice}
{nondetR}{
	\cmdc \choice \cmdd \tra{\tau} \cmdd
}
\ruledefNamed{\colwidthA}{\Pseqc}
{pseqcA}{
    \Rule{
        c_1 \tra{\aca} c_1'
    }{
        c_1 \ppseqm c_2 \tra{\aca} c_1' \ppseqm c_2
    }
}
\ruledefNamed{\colwidthB}{\Pseqc}
{pseqcC}{
		\Nil \ppseqm c_2 \tra{\tau} c_2 
}
\ruledefNamed{\columnwidth}{\Pseqc}
{pseqcB}{
	    \Rule{
        c_2 \tra{\acb} c_2'
        \qquad
        \rocmdm{\acb'}{c_1}{\acb}
    }{
        c_1 \ppseqm c_2 \tra{\acb'} c_1 \ppseqm c_2'
    }
}

\OMIT{
\ruledefNamed{88mm}{Loop}
{loop-unfold-rule}{
	\begin{array}{rl}
	&
	\WHbc 
	\\
	\tra{\tau} 
	&
	(\guard{b} \cbef \cmdc \scomp \,\, \WHbc) \choice (\guard{\neg b} \cbef \Skip)
	\end{array}
}
}

\end{centering}

\paragraph{Sequential fragment}
\labelsect{overview-rule}

The operational semantics of an instruction
is simply a step labelled by the instruction itself \refrule{action}.
After executing the corresponding step an instruction $\aca$ is terminated, \ie, becomes $\Skip$.
The semantics of loops is given by unfolding a nondeterministically-chosen finite number of times \refrule{iterate} (recall \refeqn{finiter-def}).%
\footnote{
We ignore infinite loops to avoid the usual complications they introduce,
since they do not add anything 
to the discussion of instruction reordering.
}
A nondeterministic choice (the \emph{internal choice} of CSP \cite{CSP}) can choose either branch \refrules{nondetL}{nondetR}.
A \pseqc $\cmdc_1 \ppseqm \cmdc_2$ 
can take a step if $\cmdc_1$ can take a step \refrule{pseqcA}, 
and continues with $\cmdc_2$ when $\cmdc_1$ has terminated \refrule{pseqcC},
as in standard sequential composition.
Together these rules give a standard sequential semantics for imperative programs, and we refer to them as the ``sequential fragment''.

\OMIT{
\begin{theorem}
The traces of $\csc$ are equivalent to executing $\cmdc$ with just the sequential fragment of the semantics.
\end{theorem}
\begin{proof}
Straightforward by \refmm{sc}.
\end{proof}
}

\paragraph{Parallelized sequential composition}
We now consider a further rule, unique to \impro, that allows reordering of instructions.
Rule \refrule{pseqcB}
states that given a program $\cmdc_1 \ppseqm \cmdc_2$, an instruction $\acb$ of $\cmdc_2$ can happen before the instructions
of $\cmdc_1$, provided that
 $\rocmdm{\acb'}{\cmdc_1}{\acb}$,
\ie, $\acb$ is not dependent on instructions in $\cmdc_1$ (according to the rules of model $\mm$), 
and the result of (accumulated) forwarding of instructions in $\cmdc_1$ to $\acb$ results in $\acb'$.
This is given by lifting a model $\mm$ to commands, defined inductively below (recall \refeqns{mm-def}{defn-roa}).
\labelsect{reordering-forwarding}
\labelsect{roc-fwd}
\begin{gather}
	\mm(\Nil) = \id 
	\quad
	\mm_1(c_1 \ppseq{\mm_2} c_2) = \mm_1(c_1) \comp \mm_1(c_2)
	\labeleqn{roc-pseqc}
	\\
	\mm(c_1 \nondet c_2) = \mm(c_1) \int \mm(c_2)
	\quad
	\mm(\iteratecm) = \bigintUn \mm(\finiteratecmn)
	\labeleqn{roc-choice}
	\OMIT{
	\\
	\rocmd{\acb''}{c_1 \ppseq{\mm_2} c_2}{\acb}{\mm_1}
	\ if \
	\exists \acb' \dot
		\rocmd{\acb''}{c_1}{\acb'}{\mm_1}
		\land
		\rocmd{\acb'}{c_2}{\acb}{\mm_1}
	\labeleqn{roc-pseqc}
	\\
	\rocmdm{\acb'}{c_1 \choice c_2}{\acb}
	\ if \
		\rocmdm{\acb'}{c_1}{\acb}
		\land
		\rocmdm{\acb'}{c_2}{\acb}
	\labeleqn{roc-choice}
	}
\end{gather}
Any instruction may reorder with the empty command $\Nil$.
Reordering according to memory model $\mm_1$ over $c_1 \ppseq{\mm_2} c_2$ is the relational 
composition of the orderings of $c_1$ and $c_2$ with respect to $\mm_1$ 
(independent of $\mm_2$)
\refeqn{roc-pseqc}.
\OMIT{
(although in most practical situations they will be the same).  
For instance, in the command
$
	(c \scomp d) \ppseqm e
$,
all instructions of $d$ occur after instructions of $c$, but instructions of $e$ may be interleaved with them according to $\mm$; 
the sequential composition should not have a non-local effect on other reorderings (for instance if $\mm$ was parallel
composition it should not be constrained by ordering in another ``thread'').
}
Reordering over a choice is possible only if reordering can occur over both branches \refeqn{roc-choice}
(but choices can be resolved via \refrules{nondetL}{nondetR}).
Reordering over an iteration is possible only if reordering is possible over all possible
unrollings.


\OMIT{
The rules for parallel composition \refeqn{pl-def} can be derived from \refrule{pseqc} specialised for $\PARmm$, and they reduce to the usual rules for 
interleaving parallel.
\begin{gather*}
	\Rule{
		\cmdc_1 \tra{\aca} \cmdc_1'
	}{
		\cmdc_1 \pl \cmdc_2 \tra{\aca} \cmdc_1' \pl \cmdc_2
	}
	\qquad
	\Rule{
		\cmdc_2 \tra{\aca} \cmdc_2'
	}{
		\cmdc_1 \pl \cmdc_2 \tra{\aca} \cmdc_1 \pl \cmdc_2'
	}
	\qquad
	\begin{array}{l@{~}c@{~}l}
		c \pl \Nil &\tra{\tau}& c \\
		\Nil \pl d &\tra{\tau}& d 
	\end{array}
\end{gather*}
}

\OMIT{
Reordering over iterations is nontrivial to define in the general case where an iteration of the loop may forward to an instruction.
However, iterations can be unfolded to give a finite sequence of command separated by \pseqc
according to \refrule{iterate}, and therefore we can derive the following, letting $c^*$ stand for \emph{finite} iteration of $c$.
\begin{equation}
	\rocmdm{b}{c}{b} \imp \rocmdm{b}{c^*}{b}
\end{equation}
This states that if an instruction can be reordered before a command $c$ unchanged then it can be reordered before any finite iteration of $c$.
Specific cases where instructions can be reordered with forwarding, or where instructions from one iteration can intermingle with a previous
iteration, must be handled after explicit unrolling.
}


\paragraph{Trace semantics}
\labelsect{trace-semantics}

\newcommand{\Visible}[1]{{\sf visible}~#1}
\newcommand{\Silent}[1]{{\sf silent}~#1}
\newcommand{\Infeasible}[1]{{\sf infeasible}~#1}
\newcommand{\Visa}{\Visible{\aca}}
\newcommand{\Silenta}{\Silent{\aca}}
\newcommand{\Infeasa}{\Infeasible{\aca}}

Given a program $\cmdc$ the operational semantics generates a \emph{trace}, that is, a finite sequence of steps
$\cmdc_0 \tra{\aca_1} \cmdc_1 \tra{\aca_2} \ldots$ where the labels in the trace are actions.
We write $c \xtra{t} c'$ to say that $c$ executes the actions in trace $t$ and evolves to $c'$.
\OMIT{
inductively constructed below.
The base case for the induction is given by $c \xtra{\eseq} c$.
\begin{eqnarray}
	c \tra{\aca} c' \land \Visa \land c' \xtra{t} c'' &\imp& c \xtra{\aca \cat t} c''
	\labeleqn{ptrace-visible}
	\\
	c \tra{\aca} c' \land \Silenta \land c' \xtra{t} c'' &\imp& c \xtra{t} c''
	\labeleqn{ptrace-silent}
	\\
	\also
	\Meaning{c} ~~\sdef~~  \{t | c \xtra{t} \Nil \}
	\labeleqn{def-meaning}
	\quad\qquad
	c \refsto d &\sdef & \Meaning{d} \subseteq \Meaning{c}
	\labeleqn{def-refsto}
	\quad\qquad
	c \refeq d ~~\sdef~~  c \refsto d \land d \refsto c
	\labeleqn{def-refeq}
\end{eqnarray}
}
Traces of \emph{visible} actions are accumulated into the trace,
and \emph{silent} actions (such as $\tau$) are discarded, \ie, we have a ``weak'' notion of equivalence \cite{CCS}.  
A visible action is any action with a visible effect, for instance, fences, assignments, and guards with free variables.
Silent actions include any guard which is $\True$ in any state \emph{and} contains no free variables; for instance, $\guard{0 = 0}$ is silent while
$\guard{x = x}$ is not.
A third category of actions, $\Infeasa$, includes exactly those guards $\guard{b}$ where
$b$ evaluates to $\False$ in every state.  This includes actions such as $\guard{\False}$ and $\guard{x \neq x}$.
The meaning of a command $\cmdc$ is its set of all terminating behaviours, written $\Meaning{c}$, with
behaviours containing infeasible actions being excluded from consideration.  

\paragraph{Refinement}
We take the 
usual (reverse) subset inclusion definition of refinement, 
\ie, $c \refsto d$ if every behaviour of $d$ is a behaviour of $c$; our notion of command \emph{equivalence} is refinement in both direction.
From this definition we can derive expected properties for the standard operators such as
$\Meaning{c \choice d} = \Meaning{c} \union \Meaning{d}$ 
and 
$\Meaning{c \scomp d} = \Meaning{c} \cat \Meaning{d}$ (overloading `$\cat$' to mean pairwise concatenation of sets of lists).

\OMIT{
\\
\shorteqn{0.49\columnwidth}{
	\cmdc_1 \ppseqm \cmdc_2
	\refsto
	\cmdc_1 \scomp \cmdc_2
	\label{law:keep-order}
}
\shorteqn{0.47\columnwidth}{
	c \choice d 
	\refsto
	c
	\labellaw{chooseL}
}
\OMIT{
&
\shorteqn{0.32\columnwidth}{
	\tau \ppseqm c 
	\refeq
	c
	\labellaw{elim-tau-step}
}
}
\\
}

Properties for \pseqc are, of course, more interesting, and we provide some below that
we make use of in the rest of the paper.
\begin{eqnarray}
    \cmdc_1 \ppseqm \cmdc_2
    &\refsto&
    \cmdc_1 \scomp \cmdc_2
    \label{law:keep-order}
	\\
    c \choice d
    &\refsto&
    c
    \labellaw{chooseL}
\\
	(c_1 \ppseqm c_2) \ppseqm c_3
	&\refeq&
	c_1 \ppseqm (c_2 \ppseqm c_3)
	\labellaw{pseqc-assoc}
	\\
	(\aca \scomp \cmdc) \pl \cmdd
	&\refsto&
	\aca \scomp (\cmdc \pl \cmdd)
	\label{law:fix-interleaving}
	\\
	(c_1 \choice c_2) \pl d 
	&\refeq&
	(c_1 \pl d) \choice (c_2 \pl d )
	\labellaw{dist-choice-pl}
\end{eqnarray}
\reflaw{keep-order} states that sequential composition is always a refinement of \pseqc.
\reflaw{chooseL} is the standard resolution of a choice to its left operand (a symmetric law holds for the right operand).
\Pseqc is associative by \reflaw{pseqc-assoc}, provided both instances are parameterised by the same model $\mm$.
The interleaving semantics allows us to derive 
\reflaw{fix-interleaving}, a typical interleaving law, and
\reflaw{dist-choice-pl} for distributing choice over parallel composition; these laws are important for reasoning about the effects of different 
instruction reorderings.

Now consider the case of two instructions in sequence. 
\begin{eqnarray}
	\aca \nrom \acb
	\entails \sspace \sspace
	\aca \ppseqm \acb 
	&\refeq&
	\aca \scomp \acb
	\labellaw{2actions-keep-order}
	\\
	\roabab 
	\entails \sspace \sspace
	\aca \ppseqm \acb 
	&\refsto&
	\acb' \scomp \aca
\labellaw{2actions-swap-order}
	\\
	\roabab 
	\entails \sspace \sspace
	\aca \ppseqm \acb 
	& \refeq &
	(\aca \scomp \acb) \choice (\acb' \scomp \aca)
\labellaw{2actions-reduce}
\\
	c_1 \ppseqm \ffence \ppseqm c_2 
	&\refeq&
	c_1 \scomp \ffence \scomp c_2 
	\labellaw{fence-to-seqc}
\end{eqnarray}
\reflaw{2actions-keep-order} states that if $\acb$ cannot be reordered according
to $\mm$ then they are executed sequentially.  \reflaw{2actions-swap-order} states that if reordering is allowed then that is one possible behaviour.
\reflaw{2actions-reduce}
composes these two rules to reduce \pseqc to a choice over sequential compositions,
eliminating the memory model.
Similarly,
a full fence restores order and hence sequential reasoning as in \reflaw{fence-to-seqc}.



\newcommand{\reorderings}[2]{{\sf ro}_{#2}(#1)}
\newcommand{\reorderingsc}[1]{\reorderings{\cmdc}{#1}}
\newcommand{\reorderingscm}{\reorderingsc{\mm}}

\newcommand{\pseqTrace}{\mathrel{{;}\!\!\!|}}

\newcommand{\xtraM}[2]{\xtra{#2}_{#1}}
\newcommand{\xtram}[1]{\xtraM{\mm}{#1}}
\newcommand{\taug}{t^+}

\paragraph{Monotonicity}
Monotonicity (congruence) holds for the standard operators of \impro, but
monotonicity of \pseqc contains a subtlety in that the allowed traces of $c_1 \ppseqm c_2$ are dependent on the reorderings allowed by $c_1$ with respect to $\mm$
(Rule \refrule{pseqcB}). 
To handle this we need a stronger notion of refinement, written
	$c \refstoROm c'$,
where traces are augmented to track 
the reorderings allowed,%
\footnote{
Similarly to refusal
sets in CSP's failures/divergences model \cite{Roscoe98}.
} allowing strengthening only.
\begin{theorem}
	\labellaw{mono-pseqc}
	\labellaw{pseqc-monoR}
	\labellaw{pseqc-monoL}
	\labeleqn{pseqc-monoM}
$
	c \ppseq{\mm} d 
	\refsto
	c' \ppseq{\mm'} d'
$ 
if
	$c \refstoROm c'$, $d \refsto d'$, and $\mm \roref \mm'$.
\end{theorem}
\begin{proof}
By induction on traces: the requirement for the left argument is a consequence of \refrule{pseqcB}; and strengthening of the memory model reduces
the number of possible traces.
\end{proof}


\OMIT{
Ideally we could define a compositional operator on traces such that 
$\Meaning{c \ppseqm d} = \Meaning{c} \oplus \Meaning{d}$.
However, forwarding, and hence reordering, can weaken
reordering constraints (we give an example below), and since by \reflaw{2actions-swap-order} reordering is a refinement, then reordering 
can allow more behaviours, breaking monotonicity in the left argument.
We can restore monotonicity with a stronger constraint than just trace refinement, by augmenting traces with the reordering 
constraints of the program that generated each step in the trace.
The \emph{reordering set} for a command $c$ is defined as
$
	\reorderingscm \sspace \sdef\sspace \{(\acb', \acb) | \rocmdm{\acb'}{\cmdc}{\acb} \}
$.
We inductively define a transition relation $c \xtram{\taug} c'$
where $\taug$ is an augmented trace (cf. the failures-divergences model of CSP \cite{FailuresDivergences85,vanGlabbeek}).
\OMIT{
\begin{eqnarray}
	c \tra{\aca} c' \land \Visa \land c' \xtram{\taug} c''  
	&\imp& 
	c \xtra{(\aca, \reorderingscm) \cat \taug}_\mm c''
	\labeleqn{ptraceM-visible}
	\\
	\tracesROm{c} 
	~~\sdef~~
	\{t | c \xtram{\taug} \Nil\}
	&&
	\qquad
	c \refstoROm c' 
	~~\sdef~~
	\tracesROm{c'} \subseteq \tracesROm{c}
	\labeleqn{tracesM-defs}
\end{eqnarray}

We can straightforwardly define an operator `$\pseqTrace$' so that for augmented trace $\taug_1$ and trace $t_2$, $\taug_1 \pseqTrace t_2$ is an interleaving of steps in $t_2$ with 
those of $\taug_1$ provided the reordering is allowed according to that recorded in $\taug_1$.
\begin{theorem}[Traces of \pseqc]
$\Meaning{c \ppseqm d} = \tracesROm{c} \pseqTrace \Meaning{d}$
\end{theorem}
\begin{proof}
Operator $\pseqTrace$ is defined similarly to \refrules{pseqcA}{pseqcB}.
However, arriving at a form
that gives equality was nontrivial;
see supplementary material.
\end{proof}
}
The set of augmented traces of $\cmdc$ with respect to $\mm$ is given by $\tracesROm{c}$, and we use this to define a stronger notion of refinement,
$c \refstoROm c'$.
Hence 
$c \refstoROm c'$ means that every behaviour of $c'$ is a behaviour of $c$, and that additionally
$c'$ does not allow any additional reorderings that $c$ did not allow.
Although a stronger constraint, the majority of structural rules that apply for `$\refsto$' transfer to `$\refstoROm$', \eg,
$
	c \choice d \refstoROm c
$.

We can now derive that
\pseqc is monotonic in its right-hand argument for all models with respect to trace inclusion, but needs
the stronger condition for the left-hand side.  We can also strengthen the memory model.
}

\section{Reasoning about \impro}
\labelsect{eff-semantics}

So far we have considered trace-level properties; now we turn attention to state- and predicate-based reasoning.
The action-trace semantics can be converted into a typical pairs-of-states semantics straightforwardly.
\begin{gather}
	\eff{x \asgn e} ~~ = ~~ \{(\sigma, \Update{\sigma}{x}{\evalse})\}
	\quad
	\labeleqn{effasgn-def}
	\\
	\eff{\guarde}  = \{(\sigma, \sigma) | \sigma \in e\}
	\qquad
	\eff{\fencepf} ~~ = ~~ \id
	\labeleqn{effa-def}
	\\
	\eff{\eseq} ~~=~~ \id 
	\qquad
	\quad
	\eff{a\cat t} ~~=~~ \effa \comp \eff{t} 
	\labeleqn{efft-def}
	\\
	\effc  = \bigunion \{\eff{t} | t \in \Meaning{c} \}
	\labeleqn{effc-def}
	\\
	\wpc ~~ \sdef ~~ \lambda q. \{\sigma | \forall \sigma' @ (\sigma, \sigma') \in \effc \imp \sigma' \in q \}
	\labeleqn{def-wpre}
	\\
	\htpcq ~~\sdef~~ p \imp \wpcq
	\labeleqn{def-htrip}
\end{gather}

Let the type
$\State$ be the set of total mappings from variables to values, and let the effect function $\effword: \Instr \fun \power (\State \cross \State)$ return a relation on states
given an
instruction.
We let `$\id$' be the identity relation on states, and
given a Boolean expression $e$ we write $\sigma \in e$ if $e$ is $\True$ in state $\sigma$. 
The effect of actions is straightforward from \refeqns{effasgn-def}{effa-def}, giving the trivial case
$\eff{\tau} = \id$.
The relationship with standard Plotkin style operational semantics \cite{Plotkin} is straightforward.
\begin{equation}
		c \tra{\aca} c' 
		\land
		(\sigma,\sigma') \in \effa
		\sspace \entails \sspace
		\langle c, \sigma \rangle \tra{} \langle c',\sigma' \rangle
\end{equation}
The advantage of our approach is that syntax of the action $\aca$ can be used to reason about allowed reorderings using 
\refrules{pline-commit}{pseqcB}, whereas in general one cannot reconstruct the action from a pair of states.
Mapping $\effword$ onto a trace $t$, $\mathsf{map}(\effword, t)$, yields the sequence of relations corresponding to the set of sequences of pairs of
states in a Plotkin-style trace.
We can lift $\effword$ to traces by inductively composing such a sequence of relations \refeqn{efft-def}, 
and we define the overall effect of a command by the union of the effect of its traces \refeqn{effc-def}.

The predicate transformer for weakest precondition semantics is given in \refeqn{def-wpre}.
A predicate is a set of states, so that
given a command $\cmdc$ and predicate $q$, $\wpc(q)$ returns the set of (pre) states $\sigma$ where every post-state related to $\sigma$ by $\effc$ satisfies $q$
(following, \eg, \cite{DijkstraScholten90}).
Given these definitions we can show the following.
\begin{theorem}
For \emph{Sequential} $\mm$,
$
	\wptrans{c_1 \ppseqm c_2} = \wptrans{c_1 \scomp c_2}
$
\end{theorem}
\begin{proof}
By \refdefn{seqmm}, \refth{eff-is-seq} and \refeqn{def-wpre}.
\end{proof}
We define Hoare logic judgements using weakest preconditions \refeqn{def-htrip}
(note that we deal only with partial correctness as we consider only finite traces).
From these definitions we can derive the standard rules of weakest preconditions and Hoare logic for commands such as
nondeterministic choice and sequential composition, 
but there are no general compositional rules for \pseqc.
\begin{gather}
	\htrip{\Repl{q}{x}{e}}{x \asgn e}{q} 
	\quad
	\htrip{q \land e}{\guarde}{q} 
	\quad
	\htrip{q}{\ffence}{q} 
	\labellaw{htrip-asgn}
	\\
	\htrippq{c_1 \scomp c_2} \iff 
	\htripp{c_1}{r} \land
	\htrip{r}{c_2}{q}
	\labellaw{htrip-seqc}
	\\
	\htrippq{c_1 \choice c_2} \iff 
	\htrippq{c_1} \land
	\htrippq{c_2} 
	\labellaw{htrip-choice}
	\\
	\htrippq{c_1 \ppseqm \ffence \ppseqm c_2} \iff 
	\htripp{c_1}{r} \land
	\htrip{r}{c_2}{q}
	\labellaw{htrip-pseqc-fence}
\end{gather}
\reflaw{htrip-asgn} follows from \refeqn{effa-def} and \refeqn{def-wpre}, while 
\reflaws{htrip-seqc}{htrip-choice} are straightforward by definition.
\reflaw{htrip-pseqc-fence} is a key rule that shows how, for any $\mm$ where $\ffence$
acts as a full fence, inserting a fence restores sequential reasoning in that a mid-point (predicate $r$) can be used
for compositional reasoning.

\begin{theorem}
\labelth{refsto-eff}
If $c \refsto c'$ then
$
	\eff{c'} \subseteq \effc$ ,
	$\wpc \subseteq \wptrans{c'}$,
	and 
	$
	\htpcq \imp \htrip{p}{c'}{q}
$
\end{theorem}
\begin{proof}
Straightforward from definitions.
\end{proof}

We use two key theorems to establish (or deny) properties of programs executing under weak memory models.
\begin{theorem}[Verification]
\labelth{reasoning}
\begin{gather*}
	(i)~~
	\mbox{If} \sspace
	\cmdc \refeq \undersc{\cmdc'}
	\sspace\mbox{then} \sspace
	\htpcq \iff \htrip{p}{\undersc{\cmdc'}}{q}
	\\
	(ii)~~
	\mbox{If} \sspace
	\cmdc \refsto \cmdc'
	\sspace \mbox{and} \sspace
	\htrip{p}{\cmdc'}{\neg q}
	\sspace\mbox{then}\sspace 
	\neg \htrip{p}{\cmdc}{q}
\end{gather*}
\end{theorem}
\begin{proof} 
Straightforward from definition and \refth{refsto-eff}.
\end{proof} 
\refth{reasoning}(i) is simply monotonicity of a Hoare triple, but we make the reduction to a sequential form explicit;
once this has happened standard compositional rules from Hoare logic can be applied
to establish a property of the original program.
\refth{reasoning}(ii) applies when a particular reordering of instructions in $c$ (typically $c'$ will be a sequence of instructions) contradicts some
postcondition, from which we can determine that the original program does not establish that postcondition (for a given $p$).
Hoare logic is used as the basis for reasoning about concurrent programs in the Owicki-Gries method \cite{OwickiGries76},
and so \refth{reasoning}(i) enables the application of standard techniques for concurrent programs.


We now encode some well-known memory models in our framework and show how properties and behaviours can be derived from 
the base we have provided.

\section{Total store order (TSO)}
\labelsect{tso}

The ``total store order'' memory model (used in Intel, AMD and SPARC processors; a history of its development is given in 
\cite{x86-TSO}) maintains program order on stores but is relaxed in the sense that it allows loads to come before stores.
\begin{mmdef}[$\TSOmm$]
\labelmm{tso}
$	\aca \nroT \acb $
~ except, for $x \in \SharedVars, r \in \LocalVars$, ~
\begin{gather}
	x \asgn e \roT r \asgn f
	\ltif{$x \notin \svf$ and $r \notin \fve$}
\end{gather}
\end{mmdef}
TSO allows loads to come before independent stores, and, due to forwarding, for dependent loads to ``bypass''
(recall \refeqns{defn-fve}{defn-sve}).
That is, even though by \refmm{tso} we have $x \asgn 1 \nroT r \asgn x$, due to forwarding we have
$
	\rocmd{r \asgn 1}{x \asgn 1}{r \asgn x}{\TSOmm}
$.
Note that \TSOmm allows independent register operations to also be reordered before stores.
\OMIT{
That is, since
$
	\fwd{x \smallasgn 1}{r \asgn x} = r \asgn 1
$
from \refdefn{forwarding}, we have
$
	x \asgn 1 \cbef r \asgn x
	\refsto
	r \asgn 1 \scomp x \asgn 1
	$
by \reflaw{2actions-swap-order}.
Note that the instruction type changes from a load ($r \asgn x$) to a simple update to a local register ($r \asgn 1$), and hence is not affected by any
earlier stores to $x$.  
}
We define
$\tsofence \sdef \ffence$, which is x86-TSO's primary way to restore order.%
\footnote{
x86-TSO also has store and load fences, which we discuss in the context of later memory models, but
these are effectively no-ops for TSO; however TSO's $\tsolfence$ blocks speculative execution \cite{Intel64ArchManual}.
}

\newcommand{\crod}[1]{\overset{\mathsmaller\curvearrowleft}{#1}}
\newcommand{\crodc}[1]{\crod{\cmdc_{#1}}}
\newcommand{\SB}{\T{SB}}
\newcommand{\MP}{\T{MP}}
\newcommand{\MPannot}{\MP^+}
\newcommand{\MPrc}{\undermm{\MP}{\RCmm}^+}
\newcommand{\SBtso}{\undermm{\SB}{\TSOmm}}
\newcommand{\ctson}[1]{\undermm{c_{#1}}{\TSOmm}}
\newcommand{\sbpost}{\neg (r_1 = r_2 = 0)}

\newcommand\sbactaa{x \asgn 1}
\newcommand\sbactab{r_1 \asgn y}
\newcommand\sbactba{y \asgn 1}
\newcommand\sbactbb{r_2 \asgn x}

The classic behaviours defining \TSOmm as opposed to \SCmm can be summarised by the equations below.
\begin{align}
	x \asgn 1 \ppseqt y \asgn 1 \ ~~~\refeq~~~& \ x \asgn 1 \scomp y \asgn 1
	\qquad \qquad
	\labeleqn{tso-ex1}
	\\
	x \asgn 1 \ppseqt r \asgn x ~\refsto~& \ r \asgn 1 \scomp x \asgn 1
	\labeleqn{tso-ex2}
	\\
	x \asgn 1 \ppseqt r \asgn y ~\refeq~& \ (x \asgn 1 \scomp r \asgn y) \choice (r \asgn y \scomp x \asgn 1)
	\labeleqn{tso-ex3}
\end{align}
Stores are kept in program order by \TSOmm \refeqn{tso-ex1} (an instance of \reflaw{2actions-keep-order}).  
A load of $x$ preceded by a store can use the stored value immediately \refeqn{tso-ex2} (an instance of \reflaw{2actions-swap-order}); 
only later will the store become visible to the rest of the system --
the classic \emph{bypassing} behaviour. 
A load of $y$ preceded by a store of $x$, for $x \neq y$, could be executed in either order \refeqn{tso-ex3}.
Perhaps the simplest system which can observe this behaviour is the classic ``store buffer'' ($\SB$) test
\cite{AdveHill93}.
\begin{equation*}
	\SB \sspace\sdef\sspace 
	x \asgn 1 \scomp r_1 \asgn y 
	\sspace \pl \sspace
	y \asgn 1 \scomp r_2 \asgn x
\end{equation*}
First note that in a sequential system at least one register must read the value 1.
\begin{theorem}
\labelth{htrip-SB}
$
	\htrip{x = y = 0}{\SB}{\sbpost}
$.
\end{theorem}
\begin{proof}
Lahav and Vafeaidis \cite{OG-WMMs-Lahav}
provide an Owicki-Gries proof, which we replicated using Isabelle/HOL \cite{OGinIsabelle}.
\end{proof}
However this behaviour is not ruled out under \TSOmm. 
\begin{theorem}
\labelth{SBtso}
$
	\neg \htrip{x = y = 0}{\SBtso}{\sbpost}
$.
\end{theorem}
\begin{proof}
Abbreviate
$
	c_1 \sdef \sbactab \scomp \sbactaa
$
and
$
	c_2 \sdef \sbactbb \scomp \sbactba
$, and hence
$
	\SBtso \refeq \ctson1 \pl \ctson2
$.
Also let
$
	\crodc{i} 
$
represent 
$
	c_i 
$
with its instructions reordered.
\\
\begin{tabular}{cll}
&
	\step{
		\SBtso
	\sspace
	= 
	\sspace
		\ctson1
		\pl
		\ctson2
	}
	\\

	\trans{\refeq}
	\step{
		(c_1 \choice \crodc1)
		\pl
		(c_2 \choice \crodc2)
	}
	\explanation{By \refeqn{tso-ex3}}

	\trans{\refeq}
	\step{
		(c_1 \pl c_2 )
		\choice
		(\crodc1 \pl c_2)
		\choice
		(c_1 \pl \crodc2 )
		\choice
		(\crodc1 \pl \crodc2 )
	}
	\explanation{\reflaw{dist-choice-pl}}
\end{tabular}

We have reduced $\SBtso$ to four concurrent, sequential programs representing each possible combination of reorderings.
By \reflaw{htrip-choice} we can complete the proof  by showing any one of the four violates the postcondition;
we already know that the first reordering does establish the pre/post condition by \refth{htrip-SB}, however the other three all violate it, 
as we demonstrate below for the fourth case.

\begin{tabular}{rll}
	\step{
		\crodc1 \pl \crodc2
	}
	\
	\trans{\refeq}
	\step{
		\sbactab \scomp \sbactaa
		\pl
		\sbactbb \scomp \sbactba
	}
	\explanation{Def.}

	\trans{\refsto}
	\step{
	\sbactab 
	\scomp
	\sbactbb 
	\scomp \sbactaa \scomp \sbactba
	}
	\explanation{\reflaw{fix-interleaving}}
\end{tabular}

Hoare logic (\reflaws{htrip-asgn}{htrip-seqc}), gives the following.
\[
	\htrip{x = y = 0}{\sbactab 
    \scomp
    \sbactbb 
    \scomp \sbactaa \scomp \sbactba}
	{r_1 = r_2 = 0}
\]
The proof is completed by \refth{reasoning}(ii) -- a possible reordering and interleaving contradicts the postcondition.
\end{proof}

\OMIT{
\[
	\htrip{p}{\SBtso}{q}
	\iff
	\\ \t1
	\htrip{p}{I_1}{q} \land
	\htrip{p}{I_2}{q} \land
	\htrip{p}{I_3}{q} \land
	\htrip{p}{I_4}{q} 
\]
}

\newcommand{\SBtsofenced}{\SBtso^{+\tsofence}}
\newcommand{\SBfenced}{\SB^{+\tsofence}}

To reinstate sequential behaviour under \TSOmm, fences can be inserted in both branches.
\begin{theorem}
\labelth{sbfenced}
Let $\SBfenced$ be $\SB$ with fences inserted into each branch; then
$
	\htrip{x = y = 0}{
	\SBtsofenced
	}{\sbpost}
$
\end{theorem}
\begin{proof}
By \refmm{tso}, \reflaw{htrip-pseqc-fence} and the reasoning of \refth{htrip-SB}.
\end{proof}

\OMIT{
By \reflaw{fence-to-seqc} we have
\\
\begin{tabular}{c@{}cl}
&
\step{
	\sbactaa \ppseqt \tsofence \ppseqt \sbactab
	\pl
	\sbactba \ppseqt \tsofence \ppseqt \sbactbb
}

\\

\trans{\refeq}
\step{
	\sbactaa \scomp \tsofence \scomp \sbactab
	\pl
	\sbactba \scomp \tsofence \scomp \sbactbb
}

\end{tabular}
\end{proof}
}

Note that reasoning is relatively direct in this framework: we can use properties of the model and the structure of the program to reduce reasoning to
sequential cases where established techniques can be applied (\refth{sbfenced}), or a partcular case that violates a desired property can be enumerated
(\refth{SBtso}).  Other reasoning frameworks typically monitor reorderings with respect to global abstract (graph) data structures, requiring custom assertion
languages and judgements.


\subsection{Equivalence to an explicit store buffer model}
\labelsect{tso-buf}

\newcommand{\llbrace}{\{\!|\!\!}
\newcommand{\rrbrace}{\!\!|\!\}}

\newcommand{\sbufword}{\mathbf{buf}}
\newcommand{\sbf}{\T{s}}
\newcommand{\sbcat}{\cat}
\newcommand{\sbuf}[2]{(\sbufword~#1 @ #2)}
\newcommand{\sbufs}[1]{\sbuf{\sbf}{#1}}
\newcommand{\sbufsc}{\sbufs{c}}
\newcommand{\sbufscp}{\sbufs{c'}}

\newcommand{\ebuf}{\eseq}

\newcommand{\pstorexv}{\storexv}

\renewcommand{\load}[2]{\guard{#1 = #2}}
\renewcommand{\load}[2]{#1 \asgn #2}
\renewcommand{\load}[2]{#1 \asgn #2}
\renewcommand{\loadx}[1]{\load{#1}{x}}
\newcommand{\loadxv}{\loadx{r}}
\newcommand{\loadfwdrv}{\load{r}{v}}

\newcommand{\sload}[2]{#1 \smallasgn #2}
\newcommand{\sloadx}[1]{\sload{#1}{x}}
\newcommand{\sloadxv}{\sloadx{r}}

\newcommand{\getval}[2]{fetch(#1)~from~#2}
\renewcommand{\getval}[2]{#2\llbrace#1\rrbrace}
\newcommand{\getvalx}[1]{\getval{x}{#1}}
\newcommand{\getvalxs}{\getvalx{\sbf}}

One of the best-known formalisations of a weak memory model is the operational model of x86-TSO \cite{x86-TSO,x86-TSO-TPHOLS}.  In that model the code is
executed sequentially, but interacts with
a store buffer that temporarily holds stores before sending them to the storage system, allowing loads that occur in the meantime to use values found in the buffer.  
Below we give an extension of \rowsl to add an explicit
store buffer 
$\sbf$, written $\sbufsc$,
following the semantics in \cite{x86-TSO-TPHOLS}.%
\footnote{
We give a per-process buffer, 
whereas \cite{x86-TSO-TPHOLS} uses a single global buffer, with each write in the buffer tagged by the originating process's id. 
}
\begin{align}
\labelrule{sbuf-store}
		c \ttra{x \smallasgn v} c'
	\imp&~~
		\sbufsc 
		\tra{\tau} 
		\sbuf{\sbf \cat \storexv}{c'}
	\\
\labelrule{sbuf-flush}
&~~
		\sbuf{\pstorexv \cat \sbf}{c}
		\ttra{x \smallasgn v} 
		\sbuf{\sbf}{c}
\\
\labelrule{sbuf-fence}
		c \ttra{\tsofence} c'
	\imp&~~
		\sbuf{\ebuf}{c}
		\ttra{\tsofence}
		\sbuf{\ebuf}{c'}
	\\
\labelrule{sbuf-silent}
		c \tra{\tau} c'
	\imp&~~
		\sbuf{\sbf}{c}
		\tra{\tau}
		\sbuf{\sbf}{c'}
	\\
\omit\rlap{$
		c \ttra{\sloadxv} c'
		\land
		\sbf ~=~ \sbf_1 \scat \pstorexv \scat \sbf_2
		\land
		x \notin \sbf_2
		\imp
$}
\notag
\\
	&~~
		\sbufsc \ttra{r \asgnsmall v} \sbufscp
\labelrule{sbuf-load-bypass}
\\
\omit\rlap{$
		c \ttra{\sloadxv} c'
		\land
		x \notin \sbf
		\imp
$}
\notag
\\
	&~~
		\sbufsc \ttra{\sloadxv} \sbufscp
\labelrule{sbuf-load}
\end{align}


A store is a variable/value pair, $(x,v)$,
and $x \notin \sbf$ means there is no store to $x$ in $\sbf$.
In all rules $c$ is executed using the sequential fragment of the semantics only, and we assume $x \in \SharedVars$ and $r \in \LocalVars$.
If $c$ issues 
a store, it is placed at the end of the store buffer (the system sees only a silent ($\tau$) step) \refrule{sbuf-store}.  
The first store in the buffer can be flushed to the system at any time \refrule{sbuf-flush}.
A fence can only proceed when the buffer is empty \refrule{sbuf-fence}, while
internal steps of $c$ can proceed independently of the state of the buffer \refrule{sbuf-silent}.
The interesting rules are for loads: if $c$ issues a load $\loadxv$ then this can be serviced by the buffer using the most recent value for $x$ (say $v$)
resulting in a step $r \asgn v$, and no interaction with the global system \refrule{sbuf-load-bypass}. 
If $c$ issues a load of $x$ that is not in the buffer then the load
is promoted to the system level \refrule{sbuf-load}.

\OMIT{
In a TSO system processes are structured as 
$\sbuf{\sbf}{\localssc}$ where $\sigma$ captures all local variables in $c$, that is, the buffer sits 
between the registers and the
global storage.  Hence we do not need to consider the interaction of the buffer and local variables (thus we restrict attention to stores, loads, fences and silent
steps).
}



\begin{theorem}
\labelth{sc-buf-pline} 
For any command $\cmdc$, issuing only assembler-level instructions
(stores, loads, fences and register-only operations), 
$
	\sbuf{\ebuf}{\cmdc} \refeq \plinem{\eseq}{\csc}
	$
\end{theorem}
\begin{proof}
The main difference between the store buffer semantics and a pipeline is that only stores may be fetched,
while the rules for loads combine fetching and committing in one step.  
All of the preconditions for the store buffer rules reduce to 
an equivalent form of the reordering relation over sequences of stores.
For instance, for \refrule{sbuf-load-bypass} to apply for a load $r \asgn x$, the most recently buffered (fetched) store
$(x,v)$ is used, and the promoted label is $r \asgn v$.
This is exactly the condition for a load to reorder with the equivalent trace via \refrule{pline-commit}, that is,
$
	\rocmd{r \asgn v}{t}{r \asgn x}{\TSOmm}
	$
iff trace (pipeline) $t$ is formed from stores only and is
of the form $t_1 \cat x \asgn v \cat t_2$, with $x \notin \wv{t_2}$, which follows from \refmm{tso} and lifting \TSOmm to traces as in
\refsect{pipeline}.
The other (simpler) cases similarly reduce.
\end{proof}

\begin{theorem}
\labelth{sc-buf-tso} 
For any command $\cmdc$, issuing only assembler-level instructions,
$
	\sbuf{\ebuf}{\cmdc} \refeq \ctso
	$
\end{theorem}
\begin{proof}
By \refth{sc-buf-pline} and \refth{pl=pseqc}.
\end{proof}
That is,
the semantics of \pseqc instantiated with reordering (and forwarding) given by \TSOmm gives precisely those behaviours
obtained by sequential execution with an (initially empty) store buffer.

\OMIT{
The rest of this section outlines the proof of this theorem with respect to the semantics of a store buffer;
the proof is completed in Isabelle/HOL \cite{IsabelleHOL}.
Firstly
we note that retrieving the value of a variable from a buffer $\sbf$ can be recast in terms of reordering, \ie,
we define reordering and forwarding of actions over a buffer $\sbf$
as follows, where $x \neq y$:
\begin{gather}
		\robuf{\acb}{\ebuf}{\acb} 
		\qquad \qquad
		\\
		\robuf{\acb''}{\sbf}{\acb'}
		\land
		\robuf{\acb'}{(x,v)}{\acb}
		\sspace \entails \sspace
		\robuf{\acb''}{s \scat (x,v)}{\acb}
		\labeleqn{buf-ind}
		\\
		\robuf{\loadfwdrv}{\pstorexv}{\loadxv}  
		\qquad
	\robuf{\loadxv}{\store{y}{w}}{\loadxv}	
		\labeleqn{buf-load}
\end{gather}
\refEqn{buf-ind} 
inductively defines reordering over a buffer.
\refEqn{buf-load} covers bypassing and promotion of unserviced loads, respectively.
We can then derive the following two equivalences, which cover the preconditions of the corresponding rules for loads.
\begin{gather}
	\robuf{\tau}{\sbf}{\loadxv} 
	~\iff~
	\sbf = (\sbf_1 \scat \storexv \scat \sbf_2) \land x \notin \sbf_2
	\\
	\robuf{\loadxv}{\sbf}{\loadxv} 
	~\iff~
	x \notin \sbf
\end{gather}
We further extend \refeqn{buf-load} so that $(x,v) \nroM{} \tsofence$ and $(x,v) \nroM{} y \asgn u$ 
to generalise Rules~\ref{rule:sbuf-store}-\ref{rule:sbuf-load}.
\begin{equation*}
		c \tra{\aca} c'
		\land
		\robuf{\aca'}{\sbf}{\aca}
	\imp
		\sbufs{c}
		\tra{\aca'}
		\sbufs{c}
\end{equation*}
We also note 
we may
interpret a buffer $\sbf$ as a command in \rowsl, with the straightforward semantics of a sequence of stores, \ie,
$
	\pstorexv \scat \sbf \ttra{x \smallasgn v} \sbf 
$.
Note that there is no reordering \emph{within} $\sbf$, but that other commands may be reordered before it, where we define
$\rocmdm{\acb'}{\sbf}{\acb}$ as in \refeqn{buf-ind}.
A key property is that, given the reordering rules, an intermediate store is interchangeable with the abstract buffer.
\begin{equation}
\labeleqn{tso-buf-buffer}
	\sbf \ppseqt (x \asgn v \ppseqt c) \refeq (\sbf \cat (x , v)) \ppseqt c
\end{equation}
Another key property is that the execution of a buffer command to completion (possibly leaving some writes in a non-empty buffer),
can be composed inductively.
\begin{lemma}
\labellemma{tso-buf-split}
$
	\sbufs{c_1 \ppseqs c_2} \xtra{t} \sbuf{\sbf''}{\Nil}
$
holds iff
there exists $\sbf', t_1,$ $t_2$ where $t = t_1\scat t_2$ and
$
		\sbufs{c_1} \xtra{t_1} \sbuf{\sbf'}{\Nil} 
		\quad and \quad
		\sbuf{\sbf'}{c_2} \xtra{t_2} \sbuf{\sbf''}{\Nil}
$
\end{lemma}

We can now state the main equivalence theorem. 
\begin{theorem}
\labelth{sc-buf-tso-gen}
Provided $\cmdc$ is encapsulated by a local state mapping all local variables referenced in $\cmdc$, and $\cmdc$
contains actions corresponding to atomic TSO assembler actions, \ie, 
stores,
loads,
fences, or
register-only operations (including guards), then
$
	\sbuf{\sbf}{\csc} \refeq (\sbf \ppseq{\TSOmm} \ctso)
$
\end{theorem}
}

\section{Release consistency}
\labelsect{rc-model}

\newcommand{\rcunspecmm}{\Gmm}
\newcommand{\pseqRA}{\ppseq{\RCmm}}

The \emph{release consistency} memory model \cite{ReleaseConsistency90} has been highly influential, having been implemented in the Dash processor \cite{DASH},
guided the development of the C language memory model
\cite{BoehmAdveC++Concurrency}, and the concepts incorporated into ARM \cite{ARMv8.4} and RISC-V \cite{RISC-V-ISAManual2017}.  The key concept revolves around \emph{release writes} and
\emph{acquire loads}: a release write is stereotypically used to set a flag to indicate a block of computation has ended,
and and an acquire load is correspondingly used to observe a release write.  Code before the release should happen before, and code after the acquire should
happen after; 
conceptually these are weaker (one-way) fences.
Release consistency's motivation was finding an easy-to-implement mechanism for interprocess
communication that is feasible and inexpensive computationally, and relatively straightforward for programmers.

We extend the action syntax of \impro to include 
\emph{ordering constraints} ($\oc$) as annotations 
to any action, though as noted above release store and acquire load are the most commonly used.
\begin{gather}
	\oc ~~\ttdef~~ \release \csep \acquire
	\qquad
	\aca ~~\ttdef~~ \ldots \csep \moda{\oc} 
	\labeleqn{ra-syntax}
	\\
	\fwd{\aca}{(\modb{\oc})} = \modAct{(\fwd{\aca}{\acb})}{\oc}
	\qquad
	\fwd{(\moda{\oc})}{\acb} = \fwd{\aca}{\acb}
	\labeleqn{RC:fwd}
	\OMIT{
	\\
	\relfence 
	\sdef \modR{\tau}
	\qquad
	\acqfence 
	\sdef \modA{\tau}
	\labeleqn{rc-defined-actions}
	}
\end{gather}
Forwarding for the new annotated actions is defined inductively so that the base actions take effect and ignore the annotations \refeqn{RC:fwd};
and we define $\eff{\moda{\oc}} = \eff{\aca}$.

Following \cite{ReleaseConsistency90} we distinguish two models, $\RCmm$ (where $pc$ stands for ``processor consistency'') and $\RCSCmm$
(where $sc$ stands for ``sequential consistency''), the latter of which is a strengthening of the former;
an alternative would be to distinguish $pc/sc$ in the annotations themselves, allowing mixing of the two types in one model 
(cf. ARM's \T{ldar}/\T{ldapr} instructions).
For simplicity we assume that
$\Gmm$ (\refmm{G}) controls reordering outside of annotation considerations, although in the theory of \cite{ReleaseConsistency90}
stronger constraints are possible.
\begin{mmdef}[$\RCmm$]
\labelmm{RCmm}
	$\aca \roRC \acb 
	\ltiff{
	$\aca \roM{\Gmm} \acb$
	}
	$
	except
\begin{eqnarray}
	&
	\aca \nroRC 
	\modR{\acb}
	\roRC \acc
	&
		\ltiff{$\acb \roM{\RCmm} \acc$}
	\labeleqn{RC:rel}
	\\
	&
	\aca \roRC 
	\modA{\acb}
	\nroRC \acc
	&
		\ltiff{$\aca \roM{\RCmm} \acb$}
	\labeleqn{RC:acq}
	\OMIT{
	\\
	&
	\aca 
	\roRC \acb 
	&
	\ltiff{
	$\aca \roM{\rcunspecmm} \acb$
	}
	\labeleqn{RC:default}
	}
\end{eqnarray}
\end{mmdef}
\begin{mmdef}[$\RCSCmm$]
\labelmm{RCSCmm}
	$\aca \roRCSC \acb 
	\ltiff{
	$\aca \roM{\RCmm} \acb$
	}
	$
	except
	$
	\modR{\aca} 
	\nroRCSC 
	\modA{\acb}
	$.
\OMIT{
\begin{eqnarray}
	&
	\modR{\aca} 
	\nroRCSC 
	\modA{\acb}
	&
	\labeleqn{RCSC:rel-acq}
	\\
	\quad \sspace
	&
	\aca \roRCSC \acb 
	&
	\quad
	\ltiff{
	$\aca \roM{\RCmm} \acb$
	}
	\labeleqn{RCSC:default}
\end{eqnarray}
}
\end{mmdef}
\RCmm straightforwardly follows the intuition of \cite{ReleaseConsistency90}, where a release action $\modR{\acb}$ is always blocked from reordering
and hence all earlier instructions must be complete before it can execute, but it does not block later instructions from happening early 
\refeqn{RC:rel} (provided
$\acb$ does not on its own block later instructions, calculated by recursively applying the reordering relation).
An acquire action is the converse \refeqn{RC:acq}.
\RCSCmm strengthens \RCmm by additionally requiring order between release and acquire actions in the one thread
(the reverse direction is already implied).
Consider the behaviour of the classic ``message passing'' pattern ($\MP$).
\begin{equation}
	\labeleqn{def-MP}
	\MP \sspace \sdef \sspace
	x \asgn 1 \scomp y \asgn 1
	\sspace\pl\sspace
	r_1 \asgn y \scomp r_2 \asgn x
\end{equation}
\begin{theorem}
\labelth{MPsc}
$
	\htrip{x = y = 0}{\MP}{r_1 = 1 \imp r_2 = 1}
$
\end{theorem}
\begin{proof}
Straightforward by Owicki-Gries reasoning: the stores are executed in the order $x$, $y$, and read in reverse order, hence if the latter 
is observed the former
must have taken effect.
\end{proof}
Consider using the weaker $\RCmm$ model with annotations.
\begin{equation*}
	\MPannot \sspace \sdef \sspace
	x \asgn 1 \pseqRA \modR{(y \asgn 1)}
	\sspace\pl\sspace
	\modA{(r_1 \asgn y)} \pseqRA r_2 \asgn x
\end{equation*}
Here the release annotation on the write to $y$ means that $y$
acts as a flag that $x$ has been written, and so if the other process sees the modification to $y$
via an acquire it must also see the write to $x$.
\begin{theorem}
\labelth{MPrc}
$
	\htrip{x = y = 0}{\MPannot}{r_1 = 1 \imp r_2 = 1}
$
\end{theorem}
\begin{proof} Using the definition of $\MPannot$,
~\\
\begin{tabular}{ccl}
&
	\step{
			x \asgn 1 \pseqRA \modR{(y \asgn 1)}
    		\sspace\pl\sspace
    		\modA{(r_1 \asgn y)} \pseqRA r_2 \asgn x
	}
	\\

	\trans{\refeq}
	\step{
			x \asgn 1 \scomp \modR{(y \asgn 1)}
    		\sspace\pl\sspace
    		\modA{(r_1 \asgn y)} \scomp r_2 \asgn x
	}

\end{tabular}

The equality holds by applying
\reflaw{2actions-keep-order} in each process from \refeqns{RC:rel}{RC:acq}.
Now the proof follows using the same reasoning as \refth{MPsc} (annotations have no effect on sequential semantics, only reorderings).
\end{proof}
Note that without the annotations the instructions in each process could be reordered according to $\refmm{G}$, under which conditions it is 
straightforward to find a
behaviour that contradicts $r_1 = 1 \imp r_2 = 1$.

\OMIT{
The application of \reflaw{2actions-keep-order} in each process gives the following equivalence.
\begin{equation*}
	\MP \sspace \refeq \sspace
	x \asgn 1 \scomp \modR{(y \asgn 1)}
	\sspace \pl \sspace
	\modA{(r_1 \asgn y)} \scomp r_2 \asgn x
\end{equation*}
$\MP$ has been reduced to a sequential composition ($\MP = \undersc{\MP}$), as is the intention of the release consistency model.
It is straightforward to show
$
	\htrip{x = y = 0}{\undersc{\MP}}{r_1 = 1 \imp r_2 = 1}
$
using standard methods for reasoning about concurrent sequential programs, for instance, the Owicki-Gries method \cite{OwickiGries76}
(for $\MP$ direct enumeration of the traces and their effects also gives this result).
}

\section{ARM version 8}
\labelsect{armv8}

\newcommand{\armdsb}{\T{dsb}\xspace}
\newcommand{\armdsbst}{\T{dsb.st}\xspace}
\newcommand{\armisb}{\T{isb}\xspace}


In this section we consider the latest version of ARM v8, which is simpler than earlier versions due to it being ``multi-copy atomic'' \cite{ARMv8.4}.
ARM's instruction set has artificial barriers including a ``control fence'' $\armisb \sdef \cfence$, a write barrier $\armdsbst \sdef \wwfence$, and
a full fence $\armdsb \sdef \ffence$.
\begin{mmdef}[$\ARMmm$]
\labelmm{ARMmm}
	$\aca \roA \acb 
	\ltif{$\aca \roRCSC \acb$
	}
	$
	except
\begin{eqnarray}
	&
	\aca \nroA 
	\armdsbst 
	\nroA \aca
	&
		\ltif{$\isStorea$}
	\labeleqn{A:a<wwf}
	\labeleqn{A:wwf<a}
	\\
	&
	\guardb \nroA 
	\armisb
	\nroA \aca
	&
		\ltif{$\isLoada$}
	\labeleqn{A:g<cf}
	\labeleqn{A:cf<l}
	\labeleqn{A:cf<u}
	\\
	&
	\guardb 
	\nroA
	\aca
	&
		\ltif{$\isStorea$}
	\labeleqn{A:g<u}
\end{eqnarray}
\end{mmdef}
Store fences maintain order between stores \refeqn{A:a<wwf}
(recall \refeqn{isStore}), while
control fences are blocked by branches and correspondingly block loads \refeqn{A:g<cf} (recall \refeqn{isLoad});
when taken in conjunction a control fence enforces order between loads within and before a branch, preventing
the observable effects of speculative execution.
Branches block stores, including independent stores \refeqn{A:g<u};
this is a practical consideration to do with speculating down
branches: one cannot commit stores until it is known that the branch will be taken.
Other than these exceptions, 
$\ARMmm$ behaves as $\RCSCmm$ for release/acquire annotations,%
\footnote{As mentioned in \refsect{rc-model}, ARM's \T{LDAPR} explicitly weakens the ordering between release/acquire instructions,
which can be handled by distinguishing annotations syntactically rather than within the memory model definition.
}
fundamentally behaving as $\Gmm$ (\refmm{G}).

\newcommand{\MPw}{\MP_w}
\newcommand{\MPr}{\MP_r}
\newcommand{\MPrisb}{\MP_\armisb}
\newcommand{\MPif}{\MP_{if}}

As an example of the weak nature of ARM, \ie, issuing loads before the branch condition for the load is evaluated, consider the following 
behaviour of a variant of 
the reader process of $\MP$
\refeqn{def-MP}, where the second load is guarded.
Define 
$\MPw \sdef x \asgn 1 \scomp y \asgn 1$ and 
$\MPr \sdef r_1 \asgn y \ppseqA (\If r_1 = 1 \Then r_2 \asgn x)$,
where for brevity we leave the $\ARMmm$ parameter implicit on conditionals.
\newcommand{\Mark}[1]{\underline{#1}}
\begin{theorem}
$
	\neg \htrip{x = y = 0}{\MPw \pl \MPr}{r_1 = 1 \imp r_2 = 1}
$
\end{theorem}
\begin{proof}
Consider the following behaviour of $\MPr$.
\\
\label{load-speculation}
\begin{minipage}[t]{0.47\textwidth}
\begin{tabular}{rl@{}l}
$\MPr \sdef$
&
	\step{
		r_1 \asgn y \ppseqA (\If r_1 = 1 \Then \Mark{r_2 \asgn x})
	}
	\\

	\trans{\refsto}
	\step{
		r_1 \asgn y \ppseqA \guard{r_1 = 1} \ppseqA \Mark{r_2 \asgn x} 
	}
	\explanation{\refeqn{defn-if}, \reflaw{chooseL}}

	\trans{\refsto}
	\step{
		r_1 \asgn y \ppseqA \Mark{r_2 \asgn x} \scomp \guard{r_1 = 1} 
	}
	\explanation{\reflaw{2actions-swap-order} by \refmm{G}}

	\trans{\refsto}
	\step{
		\Mark{r_2 \asgn x} \scomp r_1 \asgn y \scomp \guard{r_1 = 1} 
	}
	\explanation{\reflaw{2actions-swap-order} by \refmm{G}}
\end{tabular}
\end{minipage}
\\
The load of $x$ (underlined) may be reordered before the branch point, and subsequently before the load of $y$.  Even with the stores to $x$ and $y$ being
strictly ordered in $\MPw$ we can interleave this ordering so that the postcondition is invalidated, and complete the proof
by \refth{reasoning}(ii).
\end{proof}
Hence under \ARMmm conditionals do not guarantee sequential order.
Placing an $\armisb$ instruction inside the branch, before the second load, however, prevents this behaviour.
Define
$\MPrisb \sdef
		r_1 \asgn y \ppseqA (\If r_1 = 1 \Then \armisb \ppseqA \Mark{r_2 \asgn x})
$.
\begin{theorem}
$
	\htrip{x = y = 0}{\MPw \pl \MPrisb}{r_1 = 1 \imp r_2 = 1}
$
\end{theorem}
\begin{proof}
Consider the following behaviour of $\MPrisb$.
\\
\begin{minipage}[t]{0.999999999\columnwidth}
\begin{tabular}{cll}
&
\multicolumn{2}{l}{
	\step{
		r_1 \asgn y \ppseqA (\If r_1 = 1 \Then \armisb \ppseqA \Mark{r_2 \asgn x})
	}
}
	\\

	\trans{\refsto}
	\step{
		r_1 \asgn y \ppseqA \guard{r_1 = 1} \ppseqA \armisb \ppseqA \Mark{r_2 \asgn x} 
	}
	\explanation{\refeqn{defn-if}; \reflaw{chooseL}}

	\trans{\refeq}
	\step{
		r_1 \asgn y \scomp \guard{r_1 = 1} \scomp \armisb \scomp \Mark{r_2 \asgn x} 
	}
	\explanation{\reflaw{2actions-keep-order} by \refeqn{A:cf<l}}
\end{tabular}
\end{minipage}
\\
The loads are strictly ordered and so the proof is completed straightforwardly using OG reasoning.
\end{proof}

\OMIT{
\newcommand{\armstoreGen}[2]{#1 \asgn #2}
\newcommand{\armstore}{\armstoreGen{x}{g}}
\newcommand{\armregUGen}[2]{#1 \asgn #2}
\newcommand{\armregU}{\armregUGen{r}{e}}

\newcommand{\guardbf}{\guard{b_1}}
\newcommand{\tstore}{x \asgn e}
\newcommand{\tload}{r_1 \asgn x}
\newcommand{\treg}{r_1 \asgn e_1}
\newcommand{\tacq}{\modA{\aca}}
\newcommand{\trel}{\modR{\aca}}

\newcommand{\guardbp}{\guard{b_2}}
\newcommand{\tstorey}{y \asgn f}
\newcommand{\tloady}{r_2 \asgn y}
\newcommand{\tregp}{r_2 \asgn e_2}
\newcommand{\tacqb}{\modA{\acb}}
\newcommand{\trelb}{\modR{\acb}}

\newcommand{\noro}{\red{\cross}}
\newcommand{\ysro}{\green{\checkmark}}

\newcommand{\xnqy}{\ysro}
\newcommand{\roxb}{\ro \acb}
\newcommand{\roaa}{\aca\ro }
\newcommand{\roab}{\aca\ro\acb }

\newcommand{\fw}{^{\mbox{\guillemetleft}}}

\newcommand{\fvX}[1]{#1}
\newcommand{\rpnb}{r_2 \notin \fvX{b_1}}
\newcommand{\rnbp}{r_1 \notin \fvX{b_2}}
\newcommand{\rpne}{r_2 \notin \fvX{e}}
\newcommand{\rpnf}{r_2 \notin \fvX{e_1} \fw}
\newcommand{\rnep}{r_1 \notin \fvX{e_2}}
\newcommand{\rnfp}{r_1 \notin \fvX{f}}

\newcommand{\mybe}{*}
\newcommand{\unkn}{?}

\newcommand{\hasG}{{\mathit{ g}}}

We instantiate \refmm{ARMmm} for specific instruction types that appear in ARM assembler in \reftable{arm-table},
following the table style found in the literature, e.g., \cite{PrimerMemoryConsistency11,MM=RO+At,ArMOR}.
These are specialisations for the sort of instructions
found in assembler, \ie, stores of values, loads of shared variables into registers, and general register
operations, along with artificial barriers and release/acquire annotations.
\begin{table}[t]
\[
\hspace{5mm}
\begin{array}{l|ccc|cccc|cc}
{}_\aca \hfill {}^\acb 
          & \armdsb & \armdsbst &\armisb  &\guardbp &\tstorey &\tloady &\tregp &\tacqb &\trelb \\
\hline
\armdsb   & \noro   & \noro     &\noro    &\noro    &\noro    &\noro   &\noro  &\noro  &\noro  \\
\armdsbst & \noro   &\noro      &\noro    &\ysro    &\noro    &\ysro   &\ysro  &\roxb  &\noro  \\
\armisb   & \noro   &\noro      &\noro    &\ysro    &\ysro    &\noro   &\ysro  &\roxb  &\noro  \\
\hline
\guardbf  & \noro   &\ysro      &\noro    &\ysro    &\noro    &\rpnb   &\rpnb  &\roxb  &\noro  \\
\tstore   & \noro   &\noro      &\ysro    &\ysro    &\xnqy    &\rpne   &\rpne  &\roxb  &\noro  \\
\tload    & \noro   &\ysro      &\ysro    &\rnbp    &\rnfp    &\ysro   &\rnep  &\roxb  &\noro  \\
\treg     & \noro   &\ysro      &\ysro    &\ysro\fw &\ysro\fw &\rpnf   &\rpnf  &\roxb  &\noro  \\
\hline
\tacq     & \noro   &\noro      &\noro    &\noro    &\noro    &\noro   &\noro  &\noro  &\noro  \\
\trel     & \noro   &\roaa      &\roaa    &\roaa    &\roaa    &\roaa   &\roaa  &\noro  &\noro
\end{array}
\]
\mbox{\footnotesize
$\fw$ Forwarding may affect the reordered operation if $r_1$ appears in $\acb$.
}
\caption{
Conditions under which $\aca \protect\roA \acb$, for $\aca$ in the column and $\acb$ in the row.
Assume 
	$x,y \in \SharedVars$,  
	$r_1,r_2 \in \LocalVars$,  
	$\sve = \svf = \sv{b_1} = \sv{b_2} = \ess$,  
	$x \neq y$, $r_1 \neq r_2$
	(in cases where $x = y$ and $r_1 = r_2$ no reordering is allowed).
	To save space we use $r \protect\notin \fvX{e}$ to abbreviate $r \protect\notin \fv{e}$.
	For annotated instruction $\moda{\oc}$ we use ``$\aca \ro$'' (resp. $\ro \acb$) to indicate an inductive application of the table.
}
\labeltable{arm-table}
\end{table}
}

\subsubsection*{Conformance.}
We validate our model using litmus tests \citep{LitmusTests,UnderstandingPOWER,Mador-Haim2010,AutoSynthLitmus,SynthMMFromLitmus,TutorialARMandPOWER}.  
ARM has released an official axiomatic model using the \T{herd} tool
\cite{HerdingCats} available online via the herdtools7 application \cite{HerdARMv8} (see \cite{ARMv8A-Manual}, Sect. B2.3).
Using the \T{diy7} tool and the official model \cite{AlglaveBlog} we generated a set of 
99,881 litmus tests covering forbidden behaviours of up to 4 processes using the
instruction types covered in \refmm{ARMmm}.
In addition we used a further 5757 litmus tests covering allowed and forbidden behaviours using the tests for an earlier version of ARM \cite{HerdingCats}
and a set covering more recent features \cite{MarangetFlatResults}.
We ran these tests using the model checking tool based on the pipeline semantics in \refsect{pipeline}.
In each case (approximately 105,000 tests) \refmm{ARMmm} agreed with the published model.

\OMIT{
Fences prevent all reorderings as with TSO \refeqns{a<f}{f<a}, while a store-only barriers $\wwfence$ (corresponding to 
ARM's \T{DMB.ST} and \T{DSB.ST} instructions) maintains order on stores but not on other instruction types \refeqns{a<wwf}{wwf<a}. 
A control fence $\cfence$ prevents speculative loads when placed between a guard and a load \refeqns{g<cf}{cf<l}.  
Guards may be reordered with other guards provided they do not both access the same shared variables \refeqn{g<g} (otherwise local coherence would be violated), 
but stores to shared variables may not come before a guard evaluation \refeqn{g<u}.
This prevents speculative execution from modifying the global state, in the event that the speculation was down the wrong branch.
An update of a local variable may be reordered before a guard provided it does not affect the guard expression and respects local coherence \refeqn{g<r}.
Guards may be reordered before updates if those updates do not affect the guard expression and local coherence is respected \refeqn{u<g}.
(Note that in ARM assembler the $\loadDistinct{e}{b}$ constraints for guards are always satisfied as guards (branch points) do not reference globals.)
Assignments may be reordered as shown in \refeqn{u<u} and discussed in \refsect{overview-reordering}.
}

\subsection{An axiomatic specification}
\labelsect{arm:axiomatic}

\newcommand{\axid}[1]{[#1]}
\newcommand{\defax}[1]{\T{#1}\xspace}
\newcommand{\axpo}{\defax{po}}
\newcommand{\axpoloc}{\defax{po\text{-}loc}}
\newcommand{\axrfe}{\defax{rfe}}
\newcommand{\axfre}{\defax{fre}}
\newcommand{\axrf}{\defax{rf}}
\newcommand{\axfr}{\defax{fr}}
\newcommand{\axrfi}{\defax{rfi}}
\newcommand{\axW}{\defax{W}}
\newcommand{\axR}{\defax{R}}
\newcommand{\axacq}{\defax{\axid{A}}}
\newcommand{\axrel}{\defax{\axid{L}}}
\newcommand{\axco}{\defax{co}}
\newcommand{\axdata}{\defax{data}}
\newcommand{\axidW}{\axid{\axW}\xspace}
\newcommand{\axidR}{\axid{\axR}\xspace}
\newcommand{\axctrl}{\T{ctrl}\xspace}

\newcommand{\axisb}{\axid{\armisb}}
\newcommand{\axdsb}{\axid{\armdsb}}
\newcommand{\axdsbst}{\axid{\armdsbst}}

\newcommand{\axcomp}{;\!\!}
\newcommand{\axchoice}{~|~}

\newcommand{\axARM}{\defax{ARM}}
\newcommand{\axRC}{\defax{RC}}
\newcommand{\axob}{\defax{ob}}

\newcommand{\axkw}[1]{\texttt{\textbf{#1}}}
\newcommand{\axlet}{\axkw{let}~\xspace}
\newcommand{\axrec}{\axkw{rec}~\xspace}

\newcommand{\axirref}{\axkw{irreflexive}~\xspace}
\newcommand{\axacyc}{\axkw{acyclic}~\xspace}
\newcommand{\axas}{~\axkw{as}~\xspace}
\newcommand{\axexternal}{\axkw{external}~\xspace}
\newcommand{\axinternal}{\axkw{internal}~\xspace}

\newcommand{\labelax}[1]{\label{hcats-ax:#1}}

\newcommand{\cfrefeqn}[1]{}
\newcommand{\cfrefeqns}[2]{}
\newcommand{\cfrefeqnsc}[3]{}

\newcommand{\axmodelname}{$\defax{ax}_{ro}$\xspace}

Perhaps the best known way of describing memory models is via axioms over global traces.  
Below we give an axiomatic model using
a straightforward translation of the reordering relationship \refmm{ARMmm}.
We refer to this model as \axmodelname.
\begin{gather*}
\begin{array}{r@{~}l}
\axlet \axARM = &
	\axidW \axcomp \axpo \axcomp \axdsbst \axcomp \axpo \axcomp \axidW
    \axchoice 
	\axctrl  \axcomp \axisb
    \axchoice 
	\\ & 
    \axisb \axcomp \axpo \axcomp \axidR
    \axchoice 
	\axctrl \axcomp \axidW
	\axchoice
    \axpo \axcomp \axdsb \axcomp \axpo
\\
\axlet \axRC = &
    \axacq \axcomp \axpo
    \axchoice 
	\axpo \axcomp \axrel
    \axchoice 
    \axrel \axcomp \axpo \axcomp \axacq
\\
\axlet \axrec \axob = &
    \axARM
    \axchoice \axRC
    \axchoice \axdata \axchoice \axdata  \axcomp \axrfi
	\axchoice
	\\ & 
    \axrfe \axchoice \axfre
    \axchoice \axco
    \axchoice \axob \axcomp \axob
\\
\multicolumn{2}{l}{
\axacyc \axpoloc \axchoice \axfr \axchoice \axco \axchoice \axrf \axas \axinternal
}
\\
\also
\multicolumn{2}{l}{
\axirref \axob \axas \axexternal
}
\end{array}
\end{gather*}
An axiomatic specification is formed from relations over event traces, typically with acyclic or irreflexive constraints on the defined relations.
Union is represented by `|' and relational composition by `;'.
The \axARM relation is essentially a straight translation from \refmm{ARMmm}, along with the fence constraint from \refmm{G}.
For instance, the \axpo relation relates instructions in textual program order, and \armdsb is the set of \armdsb instructions in the program, with the
square brackets denoting
the identity relation on that set.
Hence the constraint
    ``$\axpo \axcomp \axdsb \axcomp \axpo$''
(the last in \axARM)
states a requirement that instructions before and after a fence must appear in that order in any trace.
The remaining constraints in \axARM are similarly translated, noting
\axW and \axR are the store (write) and load (read) instructions, and \axctrl relates instructions before and after a branch point.
The \axRC relation captures release/acquire constraints (\axrel/\axacq) from \refmm{RCmm}/\refmm{RCSCmm}.
The \axob relation (observation) is recursively defined to include data dependencies, including forwarding (via $\axrfi$, ``reads-from internal''), corresponding to \refmm{G0}, and the
``reads-from'' and ``from-reads'' relations, relating loads to corresponding and earlier stores, and the global coherence order $(\axco)$ on stores.
Note that these relations arise directly from our small step operational semantics.
The definition of $\axob$ and the $\axacyc$ and $\axirref$ constraints, which govern internal (local) and external (global) views of the system,
are based on the pattern of
\cite{HerdARMv8,ARMv8.4}
-- see \cite{HerdingCats} for more details on axiomatic specifications.

The \axmodelname model (available in the supplementary material) agrees with the official model \cite{HerdARMv8} on all 100,000+ litmus tests
using \T{herd7}.
Following the lead of \cite{ARMv8.4} we give a by-hand proof that the traces of the axiomatic model in \refsect{arm:axiomatic} are the same as the traces obtained by
application of the operational semantics.
\begin{theorem}
\labelth{ax=pl}
The traces of a program $c$ allowed by \axmodelname are exactly the traces of $\plinemepc$.
\end{theorem}
\begin{proof}
Consider a trace $t$ of
$\plinemepc$.  This trace must be obtained by some original sequential trace $t'$ of $c$, fetched into the pipeline via \refrule{pline-fetch}, and then
reordered by successive applications of \refrule{pline-commit}.
Without loss of generality consider the case where $t'$ is fetched into the pipeline in its entirety before any commits, and is nontrivial, \ie,
contains two or more actions.  Then the pipeline is exactly $t'$ and of the form
$\plt_1 \cat \aca \cat \plt_2 \cat \acb \cat \plt_3$, with $\aca$ occurring earlier than $\acb$ in program order in $c$.
If $\aca$ and $\acb$ are related by \axARM or \axRC then they must appear in order in any axiomatic trace of \axmodelname; 
and also they cannot be
reordered using \refrule{pline-commit} (commit), which follows from 
the straightforward relationship between \axARM \& \axRC and \refmm{ARMmm} \& \refmm{RCSCmm}, respectively.
Forwarding is covered by the \axinternal/\axexternal division in axiomatic models: the internal (local)
constraints are more strict, meaning locally sequential semantics is maintained, but externally (globally) actions may appear to occur out of order.
The fundamental constraints of \refmm{G0} and \refmm{G} are captured by
$\axdata$ and $\axpoloc$, with full fences captured by 
    $\axpo \axcomp \axdsb \axcomp \axpo$.
The $\axacyc$ and $\axirref$ constraints are implicit in an operational semantics -- a trace is always strictly ordered, and in particular loads can only access
previous stores, it is not possible to access ``future'' stores.
Hence the $\axrf$, $\axfr$ and $\axco$ constraints are implicitly enforced -- these govern the interaction between loads and stores in a trace.
In summary, the pointwise description of $\ARMmm$ translates straightforwardly to axioms over traces, where the program order ($\axpo$) 
relation captures the intervening actions.
\end{proof}


\OMIT{
\begin{theorem}
The traces of $c$ allowed by the axiomatic model in \refsect{arm:axiomatic} are exactly the traces of the process $\cunderm$ obtained from the operational semantics in
\refsect{semantics}.
\end{theorem}
\begin{proof}
By \refth{ax=pl} and \refth{pl=pseqc}.
\end{proof}
This is basically the proof, with some additional constraints about coherence and the other (standard) global orderings.  This shows the benefit of thinking this way.
}

\section{RISC-V}
\labelsect{riscv}

\newcommand{\riscfencerww}{{\tt fence~\T{rw,w}}\xspace}
\newcommand{\riscfencerrw}{{\tt fence~\T{r,rw}}\xspace}
\newcommand{\riscfencerr}{{\tt fence~\T{r,r}}\xspace}
\newcommand{\riscfenceww}{{\tt fence~\T{w,w}}\xspace}
\newcommand{\riscfencerwrw}{{\tt fence~\T{rw,rw}}\xspace}
\newcommand{\riscfencetso}{{\tt fence\T{.tso}}\xspace}
\newcommand{\riscfencei}{{\tt fence\T{.i}}\xspace}

The RISC-V memory model \cite{RISC-V-ISAManual2017,ArmstrongRISC-V} 
is influenced by ARM's weak ordering on loads and stores (corresponding to $\Gmm$), but has release consistency annotations using
the weaker $\RCmm$ (\refmm{RCmm}) rather than the stronger $\RCSCmm$ (\refmm{RCSCmm}).
It also defines six
different types of 
artificial barriers (more are technically possible but their use is not recommended \cite{ArmstrongRISC-V}):
a full fence given by
$
	\riscfencerwrw \sdef \ffence
$; 
a store fence given by
$
	\riscfenceww \sdef \wwfence
$ (identical to ARM's $\armdsbst$);
a corresponding load fence
$
	\riscfencerr
$;
two new types
$
	\riscfencerww 
$
and
$
	\riscfencerrw 
$
described below;
and a barrier used to mimic TSO's in-built weakening where loads can come before stores,
which we define as
$
	\riscfencetso \sdef \riscfencerrw \ppseqRV \riscfencerww
$.
Additionally RISC-V has
a barrier 
$
	\riscfencei 
$
which
has a technical specification beyond what is considered here, and so it is defined as a no-op ($\tau$).
\OMIT{
\begin{gather}
	\riscfencerwrw \sdef \ffence
	\qquad
	\riscfenceww \sdef \wwfence
	\qquad
	\riscfencerr \sdef \llfence
	\\
	\riscfencerww \sdef \loadgate 
	\qquad
	\riscfencerrw \sdef \storegate 
	\\
	\riscfencetso \sdef \storegate \ppseqRV \loadgate
	\qquad
	\riscfencei \sdef \tau
	\labeleqn{fence.i}
\end{gather}
}
\begin{mmdef}[$\RVmm$]
\labelmm{RVmm}
	$
	\aca 
	\roRV
	\acb 
	\ltiff{$\aca \roRC \acb$
	}
	$
	except
\begin{eqnarray}
	&
	\aca \nroA 
	\riscfencerr 
	\nroA \aca
	&
		\ltiff{$\isLoada$}
	\labeleqn{RV:a<llf}
	\labeleqn{RV:llf<a}
	\\
	&
	\aca \nroRV 
	\riscfencerww
	\roRV \acb
	&
		\ltiff{$\isLoadb$}
	\labeleqn{RV:loadgate}
	\labeleqn{RV:fencerww}
	\\
	&
	\aca \roRV 
	\riscfencerrw
	\nroRV \acb
	&
		\ltiff{$\isStorea$}
	\labeleqn{RV:storegate}
	\labeleqn{RV:storeJump}
	\labeleqn{RV:fencerrw}
	\\
	\also
	&
	\guardb \nroRV \aca
	&
		\ltiff{$\isStorea$}
	\labeleqn{RV:g<u}
\end{eqnarray}
\end{mmdef}
RISC-V's load fence, $\riscfencerr$, restricts ordering with loads \refeqn{RV:llf<a}, and is
is the straightforward dual of ARM's store fence ($\armdsbst$, \refeqn{A:wwf<a}).
RISC-V's $\riscfencerww$ barrier is intended to maintain order between loads and
stores and later stores only, allowing later loads to potentially come earlier;
it therefore allows reordering of loads, but blocks
everything else \refeqn{RV:loadgate}.
Similarly the $\riscfencerrw$ barrier ensures order between loads and later loads and stores, and hence can `jump'
over stores but is blocked by loads \refeqn{RV:storeJump}, which therefore are strictly ordered with later loads and stores.
Like ARM, RISC-V prevents stores from taking effect before branches are resolved \refeqn{RV:g<u} (see \cite{RISC-V-ISAManual2017}[Rule 11, A.3.8]).

\OMIT{
We show how these barriers take effect via the following derivation, where a $\riscfencetso$ sits between a store and a later store and load.
The intention is that the two stores are ordered (``total store order''), but the load of $z$ may come earlier.
Note that $\isStore{x \asgn 1}, \isStore{y \asgn 1},$ and $\isLoad{r \asgn z}$.
We let `$\ppseqc$' abbreviate `$\ppseqRV$'.
\begin{derivation}
	\step{
	x \asgn 1 \ppseqc \riscfencetso \ppseqc y \asgn 1 \ppseqc r \asgn z
	}

	\trans{\refeq}{Definition of $\riscfencetso$; \reflaw{pseqc-assoc}}
	\step{
		x \asgn 1 \ppseqc \riscfencerrw \ppseqc \riscfencerww \ppseqc y \asgn 1 \ppseqc r \asgn z
	}

	\trans{\refsto}{\reflaw{2actions-swap-order} by \refeqn{RV:storegate}}
	\step{
		\riscfencerrw \scomp x \asgn 1 \ppseqc \riscfencerww \ppseqc y \asgn 1 \ppseqc r \asgn z
	}

	\trans{\refsto}{\reflaw{2actions-swap-order} by \refeqn{RV:a<b}}
	\step{
		\riscfencerrw \scomp x \asgn 1 \ppseqc \riscfencerww \ppseqc r \asgn z \scomp y \asgn 1 
	}

	\trans{\refsto}{\reflaw{2actions-swap-order} by \refeqn{RV:loadgate}}
	\step{
		\riscfencerrw \scomp x \asgn 1 \ppseqc r \asgn z \scomp \riscfencerww \scomp y \asgn 1 
	}

	\trans{\refsto}{\reflaw{2actions-swap-order} by \refeqn{RV:a<b}}
	\step{
		\riscfencerrw \scomp r \asgn z \scomp x \asgn 1 \scomp \riscfencerww \scomp y \asgn 1 
	}
\end{derivation}
Here the load of $z$ has come before the store to $x$, while the store to $y$ is still prevented from coming before the store to $x$.
}

\subsubsection*{Conformance}
We tested our model against the litmus tests outlined in the RISC-V manual 
\cite{RISC-V-ISAManual2017} and made available online with expected results
\cite{RISCV-ConformanceTests}.
Restricting attention to those tests involving instructions we consider in \refmm{RVmm} (and \refmm{RCmm})
our tests agree with the official model in all 3937 cases, covering the $\RCmm$ behaviours and the six barrier types defined above, with 
\refmm{G} controlling interactions between stores and loads.

\OMIT{
Load fence litmuses:
SB+rfi-fence.r.rs.litmus
	- works; only has a load on the right
- --- parse?:
ISA02.litmus
	- doesn't auto parse  (but probably could be hacked: register name?)
OpVsHerd01.litmus
	- doesn't auto parse  (but probably could be hacked: register name?)
- --- Ignore:
C-Will01-Bad.litmus
	- has amod, doesn't auto parse  (but probably could be hacked)
SB+fence.rw.rw+ctrlfence.r.r.litmus
	- Exclusive

Not auto translated:
	- The Andy* ones are not translatable (use A and B as values)
	- The C-Will* ones use AMO
}

\section{Related work}
\labelsect{related-work}


There has been significant work in defining the semantics of processor-level instruction reordering since the 1980s
\cite{LamportSC,AccessBufferingInMultiprocs86,ShashaSnir88} and more recently under the umbrella of weak memory models
\cite{TamingRC,DecidableWmms,Boudol2009,Boudol2012,Jagadeesan2012,RGforTSO2010,CrarySullivan2015,Batty2015,Abe2016,SananSPARC2016,Kavanagh2018,DohertyVerifyingC11}.
To the best of our knowledge we are the first to encode the basis for instruction reordering as a parameter to the language, rather than as a parameter to the
semantics.  This has allowed us to describe relevant properties and relationships at the program level, 
based on per-process properties about
how individual cores manage their pipeline, and supporting program-structure-based reasoning
techniques.  In the literature weak memory model specifications are typically described with respect to properties of the global system.
The most well-used framework in this style is the axiomatic approach
\cite{SteinkeNutt04,AxiomaticPower,AutomaticallyComparingMCMs}, perhaps best exemplified by Alglave et al. \cite{HerdingCats}, in which many of the behaviours
of diverse processors such as ARM, IBM's POWER, and Intel's x86 are described, and common properties elucidated, and whose approach we compared with in
\refsect{arm:axiomatic}.
Another common approach to formalisation is with a semantics that is closer to the behaviour of a real microarchitecture,
\eg, \cite{UnderstandingPOWER,ModellingARMv8,ARMv8.4}
(we gave a direct semantic comparison to the operational model of \cite{x86-TSO,x86-TSO-TPHOLS} in \refsect{tso}).
In both the axiomatic and the concrete style
it is more difficult to derive abstract properties and to reason about a particular system over the structure of the program, 
however both give rise to efficient tools for model checking
\cite{VerifProblemAtig2010,AzizPowerMC,AzizStatelessPower,AzizTSO,MCWMMsMaude,MCforHWMMsKokologiannakis,KokoMCforC2018,AzizRA2019,AzizTSOPSpace}.

The Promising semantics \cite{PromisingSemantics,Promising-ARM,Promising2} is operational and can be instantiated with different memory models
(including software memory models), and a proof
framework has been developed for reasoning about programs.  Weak behaviours are governed by abstract global data structures that maintain
a partial order on events, and hence the semantics is defined and reasoning is performed with respect to these global data structures.  

\OMIT{
Our approach is similar to that of Colvin \& Smith
\cite{FM18}, but that paper defines only a simple prefixing command for actions, which is a special case of \pseqc.  That paper does not consider a general
theory for memory models (\refsect{overview}) nor consider pipelines, and does not address TSO, Release Consistency, or RISC-V (but does consider POWER),
and showed conformance for ARM against an older version without release/acquire atomics, against a much smaller set of litmus tests (approximatly  400 vs over
100,000 in this paper).  That theory is not machine-checked for validity, and only a few simple refinement rules are given.
}

This paper supersedes
\cite{FM18}, which defines only a simple prefixing command for actions (a special case of \pseqc).
That paper does not consider a general
theory for memory models (\refsect{overview}),
nor consider pipelines, 
and does not address TSO, Release Consistency, or RISC-V (but does consider POWER),
and showed conformance for ARM against an older version without release/acquire atomics, against a much smaller set of litmus tests (approximately  400 vs over
100,000 in this paper).  That theory is not machine-checked, contains only a few simple refinement rules, and does not employ Owicki-Gries reasoning.

The ``PipeCheck'' framework of Lustig, Martonosi et al. \cite{PipeCheck,PipeProof,TriCheck,CheckMate} is designed to validate that
processors faithfully implement their intended memory model, using a detailed pipeline semantics based on an axiomatic specification.
Given that our approach has an underlying link to the behaviour of pipelines it may be possible to extend our framework 
so that it can make use of those existing tools for processor validation.

Our operational approach based on
out-of-order instruction execution follows 
work such as Arvind et al. \cite{CRF-Arvind99,MM=RO+At,ConstructingAWMM}, and the development of the Release Consistency and related
models \cite{ReleaseConsistency90,WeakOrd-NewDef90,gharachorloo1991performance,AdveHill93,AdveBoehm2010,WMM-S}.
The algebraic approach we adopt to reducing programs is similar in style to the Concurrent Kleene Algebra \citep{CKA},
where sequential and parallel composition contribute to the event ordering.


\OMIT{

- General stuff, perhaps about complicated mms
- ---------------------
	- --- Other classification approaches?
	- Weak Memory Models: Balancing Definitional Simplicity and Implementation Flexibility
		- Not well cited, but should mention
	- Software Verification for Weak Memory via Program Transformation. Alglave et al.
		- should read this carefully, presumable prog. trans. is algebraic.  The paper at a scan is not clear
		- put this as a cite to future work, application to bigger programs?
	- Consider "Processor consistency", which appears to cover lack of mca
		- defined in  James R. Goodman. Cache consistency and sequential consistency. Technical Report no. 61, SCI Committee, March 1989.
		- considered in "Memory consistency and event ordering in scalable sharedmemory multiprocessors"
	- --- done
	- A formal hierarchy of weak memory models, Alglave 2012
		- put it at the start

- Litmus tests
- ---------------------
	- *a bunch of references for litmus tests can be found in the introduction, commented out
	- understanding power multiprocessors 2011, Sarkar et al, apparently defines the litmus naming convention?
	- Automated Synthesis of Comprehensive Memory Model Litmus Test Suites
		- nice paper for generating sets of litmus tests; should use this myself, maybe ask
	- Mador-Haim, S., Alur, R., Martin, M.: Generating litmus tests for contrasting memory consistency models - extended version. 
		- Technical report, Dept. of Computer Information Science, U. of Pennsylvania (2010)
		- more operational models, including non-store-atomic
	- Synthesizing memory models from framework sketches and Litmus tests
		- nice paper, filling gaps in sets of litmus tests

- Transforms (low level)
- ---------------------
	- Verification of an implementation of Tomasulo's algorithm by compositional model checking
		- Our pseq is an abstraction of tomasulo's concept, proved in the above paper
	- "Two techniques to enhance the performance of memory consistency models" \cite{TwoTechniques91}
		- introduces prefetching and speculative execution
		- note, weird: cited 300 times, apparently published at ICPP 91, but no record exists
	- Tomasulo designed IBM 360 FPU.
		- note that he came up with "register renaming" to remove some deps, might be useful in the shadow register discussion
		- might want to mention his "Re-order buffer" (*note: not sure that was his)
	- apparently K. C. Yeager. The MIPS R10000 Superscalar Microprocessor. IEEE Micro, pages 28–40, Apr. 1996.
		- says that speculative loads were added to SC

- ---------------------
- mentioned
- ---------------------
	- Steinke + Nutt 2004, A unif th. of shared memory consistency"
		- axiomatic, early stuff (not thread-local constraints)
	- Weak ordering—a new definition (1990) http://dl.acm.org/citation.cfm?id=325100
		- early paper from Adve + Hill (defines DRF0, superseded by DRF1 later)
	- cite: Shared memory consistency models: A tutorial \cite{SharedMemModelsTutorial}
		- cited 1400 times (peaked 2011, but still going ok)
	- cite: Memory consistency and event ordering in scalable shared-memory multiprocessors (1990)
		- cited 1600+ times (peaked 97)
		- defines "Release consistency"
	- Arvind's paper on RISC-V introduced WMM and WMM-S, which I guess should be investigated if published elsewhere
		- just list this with Arvind's other work I think
		- https://ieeexplore.ieee.org/abstract/document/8091252
	- show off that reordering is called "jockeying" in
		- MEMORY ACCESS BUFFERING IN MULTIPROCESSORS, (sec 2, p436) Michel Dubois, Chrlstoph Scheurich, Faye Briggs
		- this is a good reference to cite for early approaches to tackling the issues
		- talks about store buffers; it appears they are real things
	- Memory Models: A Case for Rethinking Parallel Languages and Hardware
		- should cite this 2010 paper by Adve + Boehm, just general stuff
		- could be mentioned with the DRF model by Arvind or whoevs
	- politic to cite Weak Ordering - A New Definition, Adve&Hill?
		- mentions "weak ordering"
			- this is the early version of their work, with DRF-0
	- Arvind+Maessen. MM = instruction reording(sic?) + store atomicity
		- close to us: they define instruction ordering very similarly (but less generally).  But they do not turn it into an operational semantics for pseqc
			- they, like everyone, tries to explain the behavious at the global level; we let them emerge from local reorderings
	- X. Shen, Arvind, and L. Rudolph. Commit-Reconcile& Fences (CRF): A New Memory Model for Architectsand Compiler Writers. 
		- this is focussed on local reorderings as well: see Fig 6, sect. 4
	- Efficient and correct execution of parallel programs that share memory 1988 (Shasha & Snir)
		- suggests allowing loads to proceed in parallel but still enforce SC
		- nice early reference wrt instruction ordering?
	- M. Frigo and V. Luchangco. Computation-centric memory models. 
		- wrt. hierarchies

- *** Done
- semantics, unknown class
- ---------------------
	- A Calculus for Relaxed Memory, Crary and Sullivan 2015 (check if this is opsem)
		- another semantics, for a new mm?
		- proves SC for a couple of "programming disciplines", may be worth checking these proofs
	- --- TSO
	- A Rely-Guarantee Proof System for x86-TSO, Ridge; check who cites it, try to find a good example
		- gives an RG system for TSO, nothing more to say
	- ** Read this: Constructing a weak memory model, Arvind, Muralidaran Vijayaraghavan, et al
		- uses the weakest mm!  Called GAM
		- uses a "hybrid" (my word) operational/axiomatic style, could be put in its own category
			- this says this is "generative", in the sense explain wmms in terms of uniprocessor optimisations, which is what we do
				- also "constructive"
			- goes into the list of early cites
	- A Denotational Semantics for SPARC TSO (Brookes)
		- better look at this, might use syntax like we do.  Otherwise doesn't have pseqc, and is TSO only

- *** Done
- Opsem
- ---------------------
	- Taming RA consistency
		- not promising??  Uses timestamps.  I guess its operational.  Cite wrt RC model, used for C
	- Atig M.F., Bouajjani A., Burckhardt S., Musuvathi M. (2012) What's Decidable about Weak Memory Models
		- this should go with the reordering table list of refs
		- also cite in the operational list with micro-arch feel
	- Relaxed operational semantics of concurrent programming languages, Boudol
		- has a semantics but no pseqc, not hugely algebraic, but not far off
	- Observation-Based Concurrent Program Logic for Relaxed Memory Consistency Models
		- another semantics, operational but no pseqc
	- Brijesh et al
	- "Relaxed Memory Models: an Operational Approach" 2009
		- has an op semantics, no pseqc
	- Brookes Is Relaxed, Almost!
		- TSO
		- adds buffers to the state in an op semantics, worth citing near our TSO encoding
	- ---
	- Promising-ARM/RISC-V: a simpler and faster operational concurrency model
		- promising, but still very hard, explicitly avoids reordering, instead keeping track of who has seen what, +at least 3 relns e.g.,
		  promises for writes

- **Done
- Model checking
- ---------------------
	- perhaps: HMC: Model Checking for Hardware Memory Models, Vaf
		- into an mc comparison section I guess, under maude
	- Context bounded analysis for POWER
		- Atig + Boujjani
	- Stateless Model Checking for POWER. Atig et al.
		- and...
	- Getting Rid of Store-Buffers in TSO Analysis, Atig, Boujjani
			- cite as MC work, not related to us really
	- Can We Efficiently Check Concurrent Programs Under Relaxed Memory Models in Maude? (2014)
		- has a picture of mm hierarchy taken from a wiki; cite in relation to maude encoding

- ** Unclear if useful
- Software mms
- ---------------------
	- Outlawing ghosts: avoiding out-of-thin-air results https://dl.acm.org/doi/10.1145/2618128.2618134 Boehm &Demsky
		- says they are bad, but necessary at pgm level?  Hard to outlaw [implicitly: in axiomatic specs] in language-level mms but is a given
		  in hardware mms, by design.  
	- "Generative operational semantics for relaxed mms"
		- This is java, no pseqc
	- could cite Making the java memory model safe (Lochbihler) for SC for DRF or Java mm
	- **Mentioned
	- Bridging the Gap between Programming Languages and Hardware Weak Memory Models; Vaf et al
		- more promising semantics
		- I guess this goes with future work, integrating with software wmms? 

- ---------------------
- not mentioned
- ---------------------
	- scan the 4 papers 3,7,10,11 in terms of processor optimisations, cited by the ": a Tutorial" paper.  Might need some more up-to-date ones
	  as well
	  	- not sure why I thoght that might be useful, ignore?
	- Neiger 2000, taxonomy of multiproc mem-ord models
		- can't find it, ignore
	- Collier 92, reasoning about pl arch
		- Can't get a copy, ignore
	- Frigo. The weakest reasonable memory model. Master’s thesis, MIT, Oct. 1997.
		- defines some properties, but we don't really want/need to go there
		- "just" a thesis
	- --- Why we need to be rigorous:
	- throwin a cite to "Relaxed memory models must be rigorous" Nardelli, Sewell, Sarkar, Alglave++
		- not actually published
	- D. Hill. Multiprocessors should support simplememory-consistency models (98)
		- makes the case that confusion of wmms doesn't justify their complexity for reasoning
		- but argues for SC which seems to be a lost cause?
}

\section{Conclusion}
\labelsect{conclusions}

In this paper we have formalised instruction-level parallelism (ILP), a feature of processors since the 1960s, and a major factor in the weak behaviours associated
with modern memory consistency models.  We showed how modern memory models build on generic properties of instruction reordering with respect to preservation of 
sequential semantics, calculated pointwise on instructions.  
We gave a straightforward abstract semantics of a pipeline, lifting pointwise comparison to sequences of instructions.
We also lifted the comparison to the command level,
and hence defined a program operator (\pseqc) which is behaviourally equivalent to using
a pipeline, but supports compositional reasoning about behaviours over the structure of parallel processes.  We empirically validated
the models for large sets of litmus tests for ARM and RISC, showed that the reordering semantics for TSO is equivalent to the established store buffer semantics,
and showed how sterotypical results emerge across a range of models, for instance, the store buffer pattern of TSO where loads can come before stores, the
message passing paradigm from release consistency using release/acquire flags to control interprocess communication, and load speculation from ARM.  We provide
in the supplementary material a model checker based on the pipeline semantics in Maude \cite{Maude,SOSMaude} and encoded and machine-checked the theory for the \impro language in Isabelle
\cite{IsabelleHOL,Paulson:94}.

The reasoning style in this framework is more direct than many in the literature, in that the nondeterminism due to ILP can be made explicit in the structure
of the program, and either shown to have no effect on desired properties, or a specific trace that contradicts a desired property can be elucidated.
In the literature reasoning about memory models is often with respect to orderings over global event traces, rather than per-process reorderings.
Our approach, \eg, to prepare for the application of the
Owicki-Gries method, appears to be simpler than techniques that incorporate memory model constraints directly \cite{OG-WMMs-Lahav,AlglaveCousot2017,DongolC11OG}.  
Furthermore, we could have instead chosen to apply 
rely/guarantee reasoning \cite{Jones-RG1,Jones-RG2} once the programs were reduced to concurrent sequential processes.
The reduction of commands involving \pseqc to a sequential form (\eg, \reflaws{2actions-reduce}{fence-to-seqc}) might be infeasible for large programs 
with few memory
barriers, but such programs are not typically where one is concerned with memory models:
more commonly memory models apply to
library data structures and
low-level applications where shared-variable communication is tightly controlled
\cite{LeWorkStealingPPoPP13,crossbeam-deque}.
An advantage of our language-level encoding
is that other program properties can be tackled in a familiar setting and at the program level, for instance, information flow logic for security or
progress properties \cite{InfoFlowWMMs,PseqcSpectre}, making use of existing frameworks directly.

\OMIT{
\Pseqc abstracts away from the details of the architecture and relates instruction reordering to the fundamental notions of sequential semantics, augmented with
artificial barriers to restore order as needed.
It allows reasoning on the structure of the program, transforming a program to reduce complex instruction orders to a choice over sequential orders, so that sequential
reasoning can be applied, or a particular reordering can be elucidated that invalidates some desired property (\refth{reasoning}).  
}

As future work we intend to extend the semantics to cover other features of modern processors that contribute to their memory models, for instance, POWER's
cache system that lacks multicopy atomicity, TSO's global locks, and ARM's global monitor for controlling load-linked/store-conditional instructions.
As these features are global they cannot be captured directly as behaviours of per-processor pipelines.
A proof technique for this extended framework will incorporate features from existing work, perhaps as a hybrid separating local reordering from global
behaviours
\cite{AlglaveVerifWMM2013,PromisingSemantics}.
Our relatively simple semantics for the pipeline contains only two stages, fetch and commit, neglecting in particular \emph{retirement}; this could be crucial
to include in security vulnerability analysis; for instance, it is the retirement of loads where invalid memory accesses are detected, and this in association with early load
speculation leads to the Meltdown vulnerability \cite{lip2018,SpectrePrime}.

\OMIT{

\robnote{Covered (drop):}
While being relatively abstract, \pseqc's operational behaviour can be related to real processors straightforwardly: when $c \ppseqm \acb$ executes $\acb$ first this
corresponds to (the instructions of) $c$ having been fetched into the pipeline and either $\acb$ has no dependencies on $\cmdc$ (according to $\mm$) or, if it does,
it still may be able to proceed using values that have been calculated in the pipeline (in, e.g., the reorder buffer for registers \cite{Tomasulo67} or a store buffer
for globals) and not yet written to registers or shared memory.  This is known as forwarding, and so the behaviour seen by other processors will be some modified
version $\acb'$ -- this is what is meant by $\rocmd{\acb'}{\cmdc}{\acb}{\mm}$, the key premise of the reordering operational semantics (\refrules{pseqcA}{pseqcB}).
In this sense we believe the \pseqc
operator, which generalises sequential and parallel composition, is a fundamental concept of assembler-level languages.  

\robnote{Martonosi: could we give a more detailed semantics of the pipeline -- more stages -- and show conformance with Martonosi semantics, and then use their tools... (future
work)}

\robnote{In future work might be nice to say retirement phase of pipeline handy for looking at meltdown exception throwing}



\robnote{Delete/move/covered already:}
Focusing on local reorderings has high explanatory power for many of the behaviours of weak memory models, 
but as noted above cannot describe all features of modern processors.
We argue that it is better to express reorderings thread-locally, and keep this 
separate from other (global) features.

In contrast to the axiomatic style over event traces,
we define the behaviours of models based on the syntax of instructions.  This allows reasoning on the
structure of processes, more easily permitting compositional reasoning.
\robnote{Move to conlcusions:}
However there are features and instruction types of processors which have a global effect that cannot be captured locally, for instance, 
POWER's cache coherence system, 
TSO's system lock \cite{x86-TSO-TPHOLS},
and in particular the behaviour of load-linked/store-conditional instructions (LL/SC), which, as described in the ARM reference manual
\cite{ARMv8A-Manual}(B2.9.2), require a global monitor to track modifications to individual locations.  
Any such behaviour cannot, and should not, be captured in the semantics of \pseqc, instead requiring a separate specification of the behaviour.
In the axiomatic style, however, such global behaviours are easier to capture;
for instance, the behaviour of LL/SC instructions in ARM is elegantly captured in one line which 
prohibits intervening writes between pairs of matched LL/SC instructions \cite{HerdARMv8}.

A survey of the extensive literature on 
formalisations
(\eg,
\cite{TamingRC,DecidableWmms,Boudol2009,Boudol2012,Jagadeesan2012,RGforTSO2010,CrarySullivan2015,Batty2015,Abe2016,SananSPARC2016,Kavanagh2018,DohertyVerifyingC11})
did not reveal any that
encode the reordering as a language-level operator.
Typically semantics frameworks are either specific to a single processor, with varying levels of abstraction,
or generic but defined as restrictions on the system as a whole.
We did not find any other work that covers the breadth of instructions we have with an extendable operational semantics that permits
algebraic reasoning, and the application of standard sequential reasoning techniques.
Essentially the technique we propose for establishing properties of programs is two-phase: understand and account for the allowed reorderings of programs, based on syntactic
constraints of how real pipelines decide on reorderings, and then apply established reasoning techniques to confirm or deny if the desired properties hold.
Most other research into reasoning can be characterised as dealing with reorderings by extending the assertion language in some way, embodying the nondeterminism due to reordering
within global data structures (such as event graphs).
In reality code that is influenced by weak memory models will be very low-level and restricted to data structure implementations from the library, and so the structures
and reorderings will be tightly controlled: either there is a sensible level of nondeterminism where indpendent instructions do not interfere, or there is some complex combination
of reordering and interleaving that results in incorrect behaviour.  Our argument is that in the former case it is better to use standard techniques, and in the latter it is
better to prove by contradiction, elucidating the actual trace, than to prove contradiction via complex data structures and novel tecniques.

\robnote{Above:}
There are effectively two contributions to compare: firstly, the semantics itself, and secondly, the reasoning technique.
The semantics is based on syntactic constraints on instructions and the model is embedded in the languge via the \pseqc operator.
As shown by the axiomatic proof, provided we are talking about local reorderings only then it is straightforward to link back.  
However pipelining is a major, but only one, feature of processors.  They need their own semantics to take into account.
On the reasoning side... (use above)

\robnote{Put this into "semantics" heading?}
Perhaps the most well-used approach to specifying real weak memory models is axiomatic, where the behaviour of the memory model is defined directly
in terms of the global
behaviour of the system, defining the stores from which loads may take their values \cite{SteinkeNutt04,AxiomaticPower,AutomaticallyComparingMCMs}.  
This approach is perhaps best exemplified by
Alglave et al. \cite{HerdingCats}, in which many of the behaviours of diverse processors such as ARM, IBM's POWER, and Intel's x86 were described, and
common properties elucidated.
In contrast to the axiomatic style over event traces,
we define the behaviours of models based on the syntax of instructions.  This allows reasoning on the
structure of processes, more easily permitting compositional reasoning.
\robnote{Move to conlcusions:}
However there are features and instruction types of processors which have a global effect that cannot be captured locally, for instance, 
POWER's cache coherence system, 
TSO's system lock \cite{x86-TSO-TPHOLS},
and in particular the behaviour of load-linked/store-conditional instructions (LL/SC), which, as described in the ARM reference manual
\cite{ARMv8A-Manual}(B2.9.2), require a global monitor to track modifications to individual locations.  
Any such behaviour cannot, and should not, be captured in the semantics of \pseqc, instead requiring a separate specification of the behaviour.
In the axiomatic style, however, such global behaviours are easier to capture;
for instance, the behaviour of LL/SC instructions in ARM is elegantly captured in one line which 
prohibits intervening writes between pairs of matched LL/SC instructions \cite{HerdARMv8}.

\OMIT{
Our semantics is presented in a conventional operational semantics style, where actions appear in the trace.  Unlike Plotkin-style operational
semantics \citep{Plotkin81,Plotkin}, we do not keep the state in the configurations, but as a first-order command of the language.  This style interacts well with
syntax-specific behaviours such as distinguishing between behaviours for registers or shared variables.  A similar approach is used by Owens \citep{OwensOcamllight}
and Abadi \& Harris \citep{AbadiHarris2009} in operational semantics; using the syntax of labels is also used in a denotational semantics by Brookes
\citep{BrookesSepLogic07}.
}

\OMIT{
We generalise an earlier simple prefix command to parallelized sequential composition, providing a more expressive and richer calculus of laws and behaviours;
We define TSO in this framework, and show that using \pseqc instantiated with \TSOmm is equivalent to executing sequential code with a store buffer;
We define Release Consistency 
We take into account a more recent version of ARM that ensures \emph{multicopy atomicity} and builds on concepts from Release Consistency,
and show conformance against a larger set of litmus test (over 5,000 in this paper vs. approx 350 in \cite{FM18})
We define RISC-V as an instance and show conformance to a set of 6000 litmus tests;
We show how to encode a weak memory model specification (a relation on instruction types) as an axiomatic model
We encode all definitions and prove all theorems in
in the Isabelle/HOL theorem prover \cite{IsabelleHOL}. 
%
}


\OMIT{
The model checking approach we developed exposed a bug in an algorithm in \citep{LeWorkStealingPPoPP13} in relation
to the placement of a control fence.  That paper includes a formal proof of the correctness of the algorithm based on the axiomatic model of
\citep{AxiomaticPower}.  The possible traces of the code were enumerated and validated against a set of conditions on adding and removing elements from the
deque (rather than with respect to an abstract specification of the deque).  As shown via derivations in \cite{FM18} the reordering is
straightforward to observe directly by looking at the code. The reordering relation for the ARM architecture show that the first control fence is redundant
(does not prevent any reorderings) because it can be reordered to the previous fence. Similarly the load in the branch can come before the branch point
itself (speculatively), and hence before the earlier load controlling the branch. The control fence, in its original position, does nothing to prevent this.
The semantics of \cite{HerdingCats} does not uncover this anomaly as directly because it is more complex to construct the whole-code relations, 
while operational
models that are more closely based on hardware mechanisms \cite{UnderstandingPOWER,ModellingARMv8} are more complex and obscure this relatively straightforward property.
}

\OMIT{
In this framework a memory model is a relationship on instructions, based on syntactic constraints (or semantics, e.g., \refmm{wpmm}).
This fits with the low-level decisions of hardware processors such as ARM and
POWER, but variable references are not in general maintained by compilers (for instance, $r \asgn y \times 0$ may be reduced to $r \asgn 0$, eliminating what is
syntactically a load).  Our main reordering principle \refeqn{reordering-principle} is based on semantic concerns: preserving sequential behaviour.  As such our
semantics may be applicable as a basis for understanding
the interplay of software memory models \cite{BoehmAdveC++Concurrency,HerdingCats,VafeiadisC11}, compiler optimisations and hardware memory models similar to \citep{BridgingPLsHMMs,RustMM}.

\robnote{Add cite to Refine19, FM19?}
}
}

\OMIT{
\begin{acks}
We thank Kirsten Winter and Ian Hayes for feedback on this work.  
We also thank Luc Maranget,
Peter Sewell, Jade Alglave, and Christopher Pulte
for assistance with litmus test analysis.
The work was supported by Australian Research Council Discovery Grant DP160102457.
\end{acks}

\robnote{Add Nick, Jeehoon.  Graeme.  Possibly thank anon reviewers of earlier and related papers}
}



\bibliography{biblio,colvinpubs,references}

\end{document}